\documentclass[12pt]{article}
\usepackage{amsmath,amsthm,amsfonts}

\begin{document}

 \def\BibTeX{{\rm B\kern-.05em{\sc i\kern-.025em b}\kern-.08em
     T\kern-.1667em\lower.7ex\hbox{E}\kern-.125emX}}
 \newtheorem{theorem}{Theorem}
 \newtheorem{lemma}{Lemma}

\def\QED{\mbox{\rule[0pt]{1.5ex}{1.5ex}}}
\def\proof{\noindent\hspace{2em}{\it Proof: }}
\def\endproof{\hspace*{\fill}~\QED\par\endtrivlist\unskip}


\def\va{{\bf a}} \def\vb{{\bf b}} \def\vc{{\bf c}} \def\vd{{\bf d}}
\def\ve{{\bf e}} \def\vf{{\bf f}} \def\vg{{\bf g}} \def\vh{{\bf h}}
\def\vi{{\bf i}} \def\vj{{\bf j}} \def\vk{{\bf k}} \def\vl{{\bf l}}
\def\vm{{\bf m}} \def\vn{{\bf n}} \def\vo{{\bf o}} \def\vp{{\bf p}}
\def\vq{{\bf q}} \def\vr{{\bf r}} \def\vs{{\bf s}} \def\vt{{\bf t}}
\def\vu{{\bf u}} \def\vv{{\bf v}} \def\vw{{\bf w}} \def\vx{{\bf x}}
\def\vy{{\bf y}} \def\vz{{\bf z}}

\def\vA{{\bf A}} \def\vB{{\bf B}} \def\vC{{\bf C}} \def\vD{{\bf D}}
\def\vE{{\bf E}} \def\vF{{\bf F}} \def\vG{{\bf G}} \def\vH{{\bf H}}
\def\vI{{\bf I}} \def\vJ{{\bf J}} \def\vK{{\bf K}} \def\vL{{\bf L}}
\def\vM{{\bf M}} \def\vN{{\bf N}} \def\vO{{\bf O}} \def\vP{{\bf P}}
\def\vQ{{\bf Q}} \def\vR{{\bf R}} \def\vS{{\bf S}} \def\vT{{\bf T}}
\def\vU{{\bf U}} \def\vV{{\bf V}} \def\vW{{\bf W}} \def\vX{{\bf X}}
\def\vY{{\bf Y}} \def\vZ{{\bf Z}}
\def\v0{{\bf 0}}

\def\btheta{{\mbox{\boldmath $\theta$}}}
\def\Beta{{\mbox{\boldmath $\eta$}}}
\def\lam{{\mbox{\boldmath $\Gamma$}}}
\def\bomega{{\mbox{\boldmath $\omega$}}}
\def\bxi{{\mbox{\boldmath $\xi$}}}
\def\brho{{\mbox{\boldmath $\rho$}}}
\def\bmu{{\mbox{\boldmath $\mu$}}}
\def\bnu{{\mbox{\boldmath $\nu$}}}
\def\btau{{\mbox{\boldmath $\tau$}}}
\def\bphi{{\mbox{\boldmath $\phi$}}}
\def\bsigma{{\mbox{\boldmath $\Sigma$}}}
\def\bLambda{{\mbox{\boldmath $\Lambda$}}}
\def\btheta{{\mbox{\boldmath $\theta$}}}
\def\bomega{{\mbox{\boldmath $\omega$}}}
\def\brho{{\mbox{\boldmath $\rho$}}}
\def\bmu{{\mbox{\boldmath $\mu$}}}
\def\bGamma{{\mbox{\boldmath $\Gamma$}}}
\def\bnu{{\mbox{\boldmath $\nu$}}}
\def\btau{{\mbox{\boldmath $\tau$}}}
\def\bphi{{\mbox{\boldmath $\phi$}}}
\def\bPhi{{\mbox{\boldmath $\Phi$}}}
\def\bxi{{\mbox{\boldmath $\xi$}}}
\def\bvarphi{{\mbox{\boldmath $\varphi$}}}
\def\bepsilon{{\mbox{\boldmath $\epsilon$}}}
\def\balpha{{\mbox{\boldmath $\alpha$}}}
\def\bvarepsilon{{\mbox{\boldmath $\varepsilon$}}}
\def\bXsi{{\mbox{\boldmath $\Xi$}}}

\def\eg{{\it e.g.}}
\def\ie{{\it i.e.}}
\def\spec{S(e^{j 2 \pi\omega}, e^{j 2\pi\nu})}
\def\tr{{\rm tr}}
\def\defines{\stackrel{{\Delta}}{{=}} }
\def\dukt{\frac {\partial \bmu^T} {\partial \theta_k}}
\def\dul{\frac {\partial \bmu} {\partial \theta_\ell}}
\def\ilam{\lam^{-1} }
\def\btheta{{\mbox{\boldmath $\theta$}}}
\def\bTheta{{\mbox{\boldmath $\Theta$}}}

\def\dlamsig{\frac {\partial \lam} {\partial \sigma^2} }
\def\dlamtheti{\frac {\partial \lam} {\partial \btheta_i} }

\def\mt{{\cal M}_2(y(n,m),\tau_n,\tau_m)}
\def\mtd{{\cal M}_2(\cdot,\tau_n,\tau_m)}
\def\drbk{\frac {\partial \lam_{{\rm PI}}} {\partial \vb_{k}}}
\def\mte{{\cal M}_2(e_i^{(\alpha ,\beta)}  (n,m) ,\tau_n,\tau_m)}
\def\mteb{{\cal M}_2(e^{(2,-1)}  (n,m) ,\tau_n,\tau_m)}
\def\wmte{{\widehat{\cal FM}}_2(e_i^{(\alpha ,\beta)}   (n,m) ;\tau_n,\tau_m, \omega, \phi)}
\def\hmte{\widehat{{\cal M}}_2(e_i^{(\alpha ,\beta)}  (n,m) ,\tau_n,\tau_m)}
\def\hmty{\widehat{{\cal M}}_2(y(n,m) ,\tau_n,\tau_m)}
\def\wmt{{\widehat{\cal FM}}_2(y(n,m) ;\tau_n,\tau_m, \omega, \phi )}
\def\dwmt{{\widehat{\cal FM}}_2(\cdot ;\tau_n,\tau_m, \omega, \phi )}
\def\fmte{{\cal FM}_2(e_i^{(\alpha ,\beta)}  (n,m) ;\tau_n,\tau_m, \omega, \phi )}
\def\fmtea{{\cal FM}_2(e_1^{(\alpha ,\beta)}  (n,m) ;\tau_n,\tau_m, \omega, \phi )}
\def\trmt{{{\cal FM}}_2(y(n,m) ;\tau_n,\tau_m, \omega, \phi )}
\def\wmtd{{{\cal FM}}_2(\cdot \ ;\tau_n,\tau_m, \omega, \phi )}
\def\tmte{{\widetilde{\cal FM}}_2(e_i^{(\alpha ,\beta)} (n,m) ;\tau_n,\tau_m, \omega, \phi)}
\def\dmte{{\cal M}_2(\cdot,\tau_n,\tau_m)}
\def\hdmte{\widehat{{\cal M}}_2(\cdot ,\tau_n,\tau_m)}
\def\hezi{\frac {1} {2}}
\def\ducl {\frac {\partial \bar \vh} {\partial \bar \vc_\ell}}
\def\drbk{\frac {\partial \lam_{{\rm PI}}} {\partial \vb_{k}}}
\def\drbone{\frac {\partial \lam_{{\rm PI}}} {\partial \vb_{1}}}
\def\drbtwo{\frac {\partial \lam_{{\rm PI}}} {\partial \vb_{2}}}
\def\drbl{\frac {\partial \lam_{{\rm PI}}} {\partial \vb_{\ell}}}
\def\dhoml{{ \rm diag} (\btau_1) \left (D_\ell \vH^R_\ell - C_\ell \vH^I_\ell  \right )}
\def\dhnul{{ \rm diag} (\btau_2) \left (D_\ell \vH^R_\ell - C_\ell \vH^I_\ell  \right )}
\def\dhomkt{\left (D_k \vH^R_k - C_k \vH^I_k  \right )^T{ \rm diag} (\btau_1) }
\def\dhnukt{\left (D_k \vH^R_k - C_k \vH^I_k  \right )^T { \rm diag} (\btau_2) }
\def\drpiabk{\frac {\partial \lam_i^{(\alpha,\beta)}} {\partial {[\bphi_i^{(\alpha,\beta)}}]_{k}}}
\def\drpjgdl{\frac {\partial \lam_j^{(\gamma,\delta)}} {\partial {[\bphi_j^{(\gamma,\delta)}}]_{\ell}}}
\def\dirsig{\frac {\partial (\vR_i^{(\alpha,\beta)})^{-1}} {\partial (\sigma_i^{(\alpha,\beta)})^2}}
\def\drsig{\frac {\partial \vR_i^{(\alpha,\beta)}} {\partial (\sigma_i^{(\alpha,\beta)})^2}}
\def\dranr{\frac {\partial (\vR_i^{(\alpha,\beta)})^{-1}} {\partial [ {\rm Re} \{ a_i^{(\alpha,\beta)}(n)\}]}}
\def\drani{\frac {\partial \vR_i^{(\alpha,\beta)}} {\partial a_i^{(\alpha,\beta)}(n)}}
\def\speci{S_i^{(\alpha,\beta)}({\rm e}^{j \omega})}
\def\drrho{\frac {\partial \vR_i^{(\alpha,\beta)}} {\partial (\rho_i^{(\alpha,\beta)})^2}}

\def\graph#1#2#3#4#5{
\begin{figure}
\vspace*{#3}
\begin{center}
\makebox[#2][l]{\special{ps: #1 x=#2 y=#3}}
\end{center}
\begin{center}
\parbox{#2}{\caption{#5}\label{#4}}
\end{center}
\end{figure}}

\newcommand{\ugraph}[5]{
\begin{figure}
\begin{center}
\makebox[#2][l]{\psfig{figure=#1,width=#2,height=#3}}
\end{center}
\begin{center}
\parbox{#2}{\caption{#5}\label{#4}}
\end{center}
\end{figure}}

\title{Strongly Consistent Model Order Selection for Estimating 2-D Sinusoids in Colored
Noise}
\author{Mark  Kliger and Joseph M. Francos
\thanks{M.  Kliger is with the Department of Electrical Engineering and Computer Science
, University of Michigan, Ann Arbor, MI, 48109-2122, USA. Tel: (734) 647-8389
FAX: (734) 763-8041, email: mkliger@umich.edu. \newline
J. M.~Francos is with the Department of Electrical and Computer
Engineering, Ben-Gurion University, Beer-Sheva 84105, Israel. Tel:
+972 8 6461842, FAX: +972 8 6472949, email: francos@ee.bgu.ac.il}
}

\date{}

\maketitle

\begin{abstract}{\bf
We consider the problem of jointly estimating the number as well as
the  parameters of two-dimensional  sinusoidal signals, observed in
the presence of an additive  colored noise field.  We begin by
elaborating on { the least squares} estimation of 2-D sinusoidal
signals, when the assumed number of sinusoids is incorrect. In the
case where the number of sinusoidal signals is under-estimated we
show the almost sure convergence of the least squares estimates to
the parameters of the dominant sinusoids. In the case where this
number is over-estimated, the estimated parameter vector obtained by
the least squares estimator contains a sub-vector that converges
almost surely to the correct parameters of the sinusoids. Based on
these results, we prove the strong consistency of a new model order
selection rule.}
\end{abstract}


\vspace{2in}


\pagebreak
 \setcounter{page}{1}

%
%

\section{Introduction}

We consider the problem of jointly estimating the number as well as
the  parameters of two-dimensional  sinusoidal signals, observed in
the presence of an additive noise field. This problem is, in fact, a
special case of a much more general problem, \cite{evcrb}: From the
2-D Wold-like decomposition  we have that any 2-D regular and
homogeneous discrete random field can be represented as a sum of two
mutually orthogonal components: a purely-indeterministic field and a
deterministic one. In this paper we consider the special case where
the deterministic component consists of a finite (unknown) number of
sinusoidal components, while the purely-indeterministic component is
an infinite order non-symmetrical half plane, (or a quarter-plane),
moving average field. This modeling and estimation problem has
fundamental theoretical importance, as well as various applications
in texture estimation of images (see, e.g., \cite{francos2} and the
references therein) and  in  wave propagation problems (see, e.g.,
\cite{ward} and the references therein).

Many algorithms have been devised to estimate the parameters of
sinusoids observed in white noise and only a small fraction of the
derived methods has been extended to the case where the noise field
is colored (see, \eg, Francos {\it et. al.} \cite{francos1}, He
\cite{He1}, Kundu and Nandi \cite{kundu}, Li and Stoica
\cite{stoica2}, Zhang and Mandrekar \cite{zhang}, and the references
therein). Most of these assume the number of sinusoids is {\it
a-priori} known. However this assumption does not always hold in
practice. In the past three decades  the problem of model order
selection for 1-D signals has received considerable attention. In
general, model order selection rules are based (directly or
indirectly) on three
popular criteria: Akaike information criterion (AIC), 
the minimum description length (MDL), 
and the maximum a-posteriori probability criterion (MAP).
All these criteria have a common form composed of two terms: a data
term and a penalty term, where the data term is the log-likelihood
function evaluated for the assumed model. The problem of modelling
multidimensional fields has received much less attention. In
\cite{map_paper03}, a MAP model order selection criterion for
jointly estimating the number and the parameters of two-dimensional
sinusoids observed in the presence of an additive white Gaussian
noise field, is derived. In \cite{map_consis}, we proved the strong
consistency of a large family of model order selection rules, which
includes the MAP based rule in \cite{map_paper03} as a special case.

In this paper we derive a strongly consistent model order selection
rule, for  jointly estimating the number of sinusoidal components
and their parameters in the presence of colored noise. This
derivation extends the results of \cite{map_consis} to the case
where the additive noise is colored, modeled by an infinite order
non-symmetrical half-plane or quarter-plane moving average
representation, such that the noise field is not necessarily
Gaussian. To the best of our knowledge this is the most general
result available in the area of model-order selection rules of 2-D
random fields with mixed spectrum.


The proposed criterion has the usual form of a data term and a
penalty term, where the first is the {\it least squares estimator}
evaluated for the assumed model order and the latter is proportional
to the logarithm of the data size.

Since we evaluate the data term for any assumed model order,
including incorrect ones, we should consider the problem of least
squares estimation of the parameters of 2-D sinusoidal signals when
the assumed number of sinusoids is incorrect. Let $P$ denote the
number of sinusoidal signals in the observed field and let $k$
denote their assumed number. In the case where the number of
sinusoidal signals is under-estimated, \ie, $k < P$,  we prove the
almost sure convergence of the least squares estimates to the
parameters of the $k$ dominant sinusoids. In the case where the
number of sinusoidal signals is over-estimated, \ie, $k > P$,  we
prove the almost sure convergence of the estimates obtained by the
least squares estimator to the parameters of the $P$ sinusoids in
the observed field. The additional $k-P$ components assumed to
exist, are assigned by the least squares estimator to the dominant
components of the periodogram of the noise field.

Finally, using this result, we prove the strong consistency of a new
model order selection criterion and show how different assumptions
regarding a noise field parameters affect the penalty term of the
criterion. The proposed criterion completely generalized the
previous results \cite{map_paper03}, \cite{map_consis}, and provides
a strongly consistent estimator of the number as well as of the
parameters of the sinusoidal components.

%
%

\section{Notations, Definitions and Assumptions}

Let $\{y(n,m)\}$ be a real valued field,
\begin{equation}
y(n,m) = \sum\limits_{i = 1}^P \rho_i^0 \cos (\omega _i^0 n +
\upsilon _i^0 m+\varphi_i^0) + w(n,m), \label{ee1}
\end{equation}
where  $0 \le n \le N - 1$, $0 \le m \le M - 1$ and for each $i$,
$\rho_i^0$ is non-zero. Due to physical considerations it is further
assumed that for each $i$, $|\rho_i^0|$ is bounded .

Recall that the {\it non-symmetrical half-plan total-order} is
defined by
\begin{eqnarray}
(i,j)\succeq (s,t) \ \mbox{iff} \ \ (i,j)\in \left\{(k,l)|k=s,l\geq
t\right\}\cup \left\{(k,l)|k>s,-\infty \leq l\leq \infty\right\}.
\end{eqnarray}

Let $D$ be an \emph{infinite} order non-symmetrical half-plane
support, defined by
\begin{equation}
D=\left\{(i,j)\in \mathbb{Z}^2: i=0, 0 \leq j \leq \infty
\right\}\cup \left\{(i,j)\in \mathbb{Z}^2: 0<i\leq \infty, -\infty
\leq j \leq \infty \right\}. \label{d}
\end{equation}
Hence the notations $(r,s) \in D$ and $(r,s) \succeq (0,0)$ are
equivalent.

We assume that $\{w(n,m)\}$
 is an infinite order non-symmetrical half-plane MA noise field, \ie,
\begin{equation}
w(n,m) = \sum_{(r,s) \in D}a(r,s) u(n-r,m-s),  \label{e3}
\end{equation}
such that the following assumptions are satisfied:
%
%

\textbf{ Assumption 1:}   The field $\{u(n,m)\}$ is an i.i.d. real
valued zero-mean random field with finite variance $\sigma^2$, such
that ${ E}[|u(n,m)|^\alpha]< \infty $ for some $\alpha>3$ .

\textbf{ Assumption 2:} The sequence ${a(i,j)}$ is an absolutely
summable deterministic sequence, \ie,
\begin{equation}
\sum_{(r,s) \in D}|a(r,s)| < \infty.
\end{equation}
%

Let $f_w(\omega,\upsilon)$ denote the spectral density function of
the noise field $\{w(n,m)\}$. Hence,
\begin{equation}
f_w(\omega,\upsilon)=\sigma^2 \bigg |\sum_{(r,s) \in D} a(r,s)e^{
j(\omega r + \upsilon s)}\bigg|^2.
\end{equation}

\textbf{ Assumption 3:} The spatial frequencies $(\omega _i^0
,\upsilon _i^0 ) \in (0,2\pi )\times (0,2\pi )$, $1\le i \le P$ are
pairwise different. In other words, $\omega _i^0 \ne \omega _j^0 $
or $\upsilon _i^0 \ne \upsilon _j^0 $, when $i \ne j$.

Let $\{\Psi_i\}$ be a sequence of rectangles such that
$\Psi_i=\{(n,m)\in \mathbb{Z}^2\mid 0 \le n \le N_i - 1, 0 \le m \le
M_i - 1\}$.

{\bf Definition 1}:  The sequence of subsets $\{\Psi_i\}$ is said to
tend to infinity (we adopt the notation $\Psi_i \rightarrow\infty $)
as $i\rightarrow\infty$  if
$$\lim\limits_{i\rightarrow\infty}\min(N_i,M_i)=\infty,$$ and
$$0<\lim\limits_{i\rightarrow\infty}(N_i/M_i)<\infty.$$ To simplify
notations, we shall omit in the following the subscript $i$. Thus,
the notation $\Psi(N,M) \rightarrow\infty $ implies that both $N$
and $M$ tend to infinity as functions of $i$, and at roughly the
same rate.

{\bf Definition 2}: Let $\Theta_k$ be a bounded and closed subset of
the $4k$ dimensional space $ \mathbb{R}^{k}\times ((0,2\pi )\times
(0,2\pi ))^k\times[0,2\pi)^k$ where for any vector $\theta _k =
(\rho_1,\omega _1 ,\upsilon _1 ,\varphi_1,\ldots ,\rho_k ,\omega _k
,\upsilon _k,\varphi_k )\in \Theta_k$ the coordinate $\rho_i$ is
non-zero and bounded for every $1\le i \le k$ while the pairs
$(\omega _i ,\upsilon _i)$ are pairwise different, so that no two
regressors coincide. We shall refer to $\Theta_k$ as the
{\textit{parameter space}}.

From the model definition (\ref{ee1}) and the above assumptions it
is clear that $$\theta _k^0 = (\rho_1^0,\omega _1^0 ,\upsilon _1^0
,\varphi_1^0,\ldots ,\rho_k^0 ,\omega _k^0 ,\upsilon
_k^0,\varphi_k^0 )\in \Theta_k.$$

Define the loss function due to the error of the $k$-th order
regression model
\begin{equation}
{\cal L}_k (\theta_k)= \frac{1}{NM}\sum\limits_{n = 0}^{N - 1}
\sum\limits_{m=0}^{M - 1} \bigg( y(n,m) - \sum\limits_{i = 1}^k
\rho_i^0 \cos (\omega _i^0 n + \upsilon _i^0 m+\varphi_i^0)\bigg)^2
. \label{lsf}
\end{equation}

A vector $\hat{\theta _k} \in \Theta_k$ that minimizes  ${\cal L}_k
(\theta_k)$ is called the {\textit{Least Square Estimate}} (LSE). In
the case where $k=P$, the LSE is a {\it strongly consistent}
estimator of $\theta_P^0$ (see, \eg, \cite{kundu} and the references
therein).

%
%

\section {Strong Consistency of the Over- and Under-Determined
LSE}

In the following subsections we establish the strong consistency of
this LSE when the number of sinusoids is under-estimated, or
over-estimated. The first theorem establishes the strong consistency
of the least squares estimator in the case where the number of the
regressors is lower than the actual number of sinusoids. The second
theorem establishes the strong consistency of the least squares
estimator in the case where the number of the regressors is higher
than the actual number of sinusoids.

%
%

\subsection{Consistency of the LSE for an
Under-Estimated Model Order}

Let $k$ denote the assumed number of observed 2-D sinusoids, where
$k<P$. For any $\delta>0$, define the set $\Delta _\delta $ to be
a subset of the parameter space $\Theta_k $ such that each vector
$\theta _k \in \Delta _\delta$ is different from the vector
$\theta _k^0 $ by at least $\delta$, at least in one of its
coordinates, \ie,

\begin{equation}
\Delta _\delta = \left[ {\bigcup\limits_{i = 1}^k {{\cal
R}_{i\delta } } } \right]\cup\left[ {\bigcup\limits_{i = 1}^k
{\Phi_{i\delta } } } \right] \cup \left[ {\bigcup\limits_{i = 1}^k
{W_{i\delta } } } \right] \cup \left[ {\bigcup\limits_{i = 1}^k
{V_{i\delta } } } \right] \ ,
\end{equation}
where
\begin{eqnarray}
&&{\cal R}_{i\delta } = \left\{ {\theta _k \in \Theta_k :\;\vert
\rho_i - \rho_i^0 \vert \ge \delta ;\delta > 0\;} \right\} \ ,
\nonumber \\
&&\Phi_{i\delta } = \left\{ {\theta _k \in \Theta_k :\;\vert
\varphi_i - \varphi_i^0 \vert \ge \delta ;\delta > 0\;} \right\} \ ,
\nonumber \\
&&W_{i\delta } = \left\{ {\theta _k \in \Theta_k :\;\vert \omega _i
- \omega _i^0 \vert \ge \delta ;\delta > 0\;} \right\} \ ,
\nonumber \\
&&V_{i\delta } = \left\{ {\theta _k \in \Theta_k :\;\vert \upsilon
_i - \upsilon _i^0 \vert \ge \delta ;\delta > 0\;} \right\} \ .
\end{eqnarray}

To prove the main result of this section we shall need an
additional assumption and the following lemmas:

\textbf{ Assumption 4:} For convenience, and without loss of
generality, we assume that the sinusoids are indexed according to a
descending order of their amplitudes, \ie,
\begin{equation}
 \rho_1^0\ge \rho_2^0 \ge \ldots \rho_k^0
 > \rho_{k + 1}^0 \ldots \ge \rho_P^0
> 0 \ ,
\end{equation}
where we assume that for a given $k$, $ \rho_k^0  > \rho_{k + 1}^0
$ to avoid trivial ambiguities resulting from the case where the
$k$-th dominant component is not unique.

\begin{lemma}
\begin{equation}
\mathop {\liminf}\limits_{\Psi (N,M) \to \infty } \mathop {\inf
}\limits_{\theta _k \in \Delta _\delta } \left( {{\cal L}_k(\theta
_k ) - {\cal L}_k(\theta _k^0 )} \right) > 0 \ \ a.s. \label{l0}
\end{equation}
\end{lemma}
\begin{proof}
See Appendix A for the proof.
\end{proof}

\begin{lemma}
\label{Wu} Let $\{ x_n,n \ge 1 \}$ be a sequence of random
variables. Then
\begin{equation}
\Pr \{x_n \le 0\ i.o. \} \le \Pr\{\mathop {\liminf }\limits_{n \to
\infty} x_n \le 0\},
\end{equation}
\end{lemma}
where the abbreviation {\it i.o.} stands for {\it infinitely
often}.

\begin{proof}
See Appendix B for the proof.
\end{proof}

The next theorem establishes the strong consistency of the least
squares estimator in the case where the number of the regressors is
lower than the actual number of sinusoids.

\begin{theorem}
Let  Assumptions 1-4 be satisfied. Then, the $k$-regressor parameter
vector $\hat {\theta }_k = (\hat{\rho}_1,\hat{\omega} _1
,\hat{\upsilon} _1 ,\hat{\varphi}_1,\ldots ,\hat{\rho}_k
,\hat{\omega} _k ,\hat{\upsilon} _k,\hat{\varphi}_k )$ that
minimizes (\ref{lsf}) is a {\it strongly consistent} estimator of
$\theta _k^0 = (\rho_1^0,\omega _1^0 ,\upsilon _1^0
,\varphi_1^0,\ldots ,\rho_k^0 ,\omega _k^0 ,\upsilon
_k^0,\varphi_k^0 )$ as $\Psi (N,M) \to \infty $.
 That is,
\begin{equation}
\hat {\theta }_k \to \theta _k^0 \enskip a.s.\enskip as\enskip \Psi
(N,M) \to \infty .
\end{equation}
\end{theorem}
\begin{proof}
The proof follows an argument proposed by Wu \cite{wu}, Lemma 1. Let
\newline $\hat {\theta }_k = (\hat{\rho}_1,\hat{\omega} _1 ,\hat{\upsilon}
_1 ,\hat{\varphi}_1,\ldots ,\hat{\rho_k} ,\hat{\omega} _k
,\hat{\upsilon} _k,\hat{\varphi}_k )$ be a parameter vector that
minimizes (\ref{lsf}). Assume that the proposition $\hat {\theta }_k
\to \theta _k^0 \enskip a.s.\enskip as\enskip \Psi (N,M) \to \infty
$ is not true. Then, there exists some $\delta>0$, such that
(\cite{chung}, Theorem 4.2.2, p. 69),
\begin{equation}
\Pr (\hat {\theta }_k \in \Delta _\delta \enskip i.o.) > 0.
\end{equation}
This inequality together with the definition of $\hat {\theta }_k$
as a vector that minimizes ${\cal L}_k$ implies
\begin{equation}
\Pr(\mathop {\inf }\limits_{\theta _k \in \Delta _\delta } \left(
{\cal L}_k(\theta _k ) - {\cal L}_k(\theta _k^0 ) \right)  \le 0
\enskip i.o.)>0.
\end{equation}

Using Lemma \ref{Wu} we obtain
\begin{equation}
\Pr(\mathop {\liminf}\limits_{\Psi (N,M) \to \infty } \mathop
{\inf }\limits_{\theta _k \in \Delta _\delta } \left( {{\cal
L}_k(\theta _k ) - {\cal L}_k(\theta _k^0 )} \right) \le 0 )\ge
\Pr(\mathop {\inf }\limits_{\theta _k \in \Delta _\delta } \left(
{{\cal L}_k(\theta _k ) - {\cal L}_k(\theta _k^0 )} \right) \le 0
\enskip i.o.)>0,
\end{equation}
which contradicts (\ref{l0}). Hence,
\begin{equation}
\hat {\theta }_k \to \theta _k^0 \enskip a.s.\enskip as\enskip
\Psi (N,M) \to \infty.
\end{equation}
\end{proof}

{\bf Remark}: Lemma 1 and Theorem 1 remain valid even under less
restrictive assumptions regarding the noise field $\{w(n,m)\}$. If
the field $\{u(n,m)\}$ is an i.i.d. real valued zero-mean random
field with finite variance $\sigma^2$, and the sequence ${a(i,j)}$
is a square summable deterministic sequence, \ie, $\sum_{(r,s) \in
D}a^2(r,s) < \infty$, then Lemma 1 and Theorem 1 hold.

%
%

\subsection{Consistency of the LSE for an Over-Estimated Model
Order}

Let $k$ denote the assumed number of observed 2-D sinusoids, where
$k>P$. Without loss of generality, we can assume that $k = P + 1$,
(as the proof for $k\ge P+1$ follows immediately by repeating the
same arguments). Let the periodogram (scaled by a factor of 2) of
the field $\{w(n,m)\}$ be given by
\begin{equation}
I_w (\omega ,\upsilon ) = \frac{2}{NM}\left| {\sum\limits_{n =
0}^{N - 1} {\sum\limits_{m = 0}^{M - 1} {w(n,m)e^{ - j(n\omega +
m\upsilon )}} } } \right|^2.
\end{equation}
The parameter spaces $\Theta _{P } $, $\Theta _{P + 1} $ are defined
as in Definition 2.

\begin{theorem}
Let  Assumptions 1-4 be satisfied. Then, the parameter vector
\newline $\hat {\theta }_{P + 1} = (\hat{\rho}_1,\hat{\omega} _1
,\hat{\upsilon} _1 ,\hat{\varphi}_1,\ldots ,\hat{\rho}_P
,\hat{\omega} _P ,\hat{\upsilon} _P,\hat{\varphi}_P,\hat{\rho}_{P+1}
,\hat{\omega} _{P+1} ,\hat{\upsilon} _{P+1},\hat{\varphi}_{P+1} )
\in \Theta _{P + 1} $ that minimizes (\ref{lsf}) with $k = P + 1$
regressors as $\Psi (N,M) \to \infty $ is composed of the vector
$\hat {\theta }_P = (\hat{\rho}_1,\hat{\omega} _1 ,\hat{\upsilon} _1
,\hat{\varphi}_1,\ldots ,\hat{\rho}_P ,\hat{\omega} _P
,\hat{\upsilon} _P,\hat{\varphi}_P )$ which is a {\it strongly
consistent} estimator of $\theta _P^0 = (\rho_1^0,\omega _1^0
,\upsilon _1^0 ,\varphi_1^0,\ldots ,\rho_P^0 ,\omega _P^0 ,\upsilon
_P^0,\varphi_P^0 )$ as $\Psi (N,M) \to \infty $;
of the pair of spatial frequencies $(\hat {\omega }_{P + 1} ,\hat
{\upsilon }_{P + 1} )$ that maximizes the periodogram of the
observed realization of the field $\{w(n,m)\}$, \ie,
\begin{equation}
(\hat {\omega }_{P + 1} ,\hat {\upsilon }_{P + 1} ) = \mathop {\arg
\max } \limits_{(\omega ,\upsilon ) \in {(0,2\pi)^2}} I_w (\omega
,\upsilon ),
\end{equation}
and of the element $\hat {\rho}_{P + 1} $ that satisfies
\begin{equation}
\hat {\rho}_{P + 1}^2 = \frac{2}{NM}I_w (\hat {\omega }_{P + 1}
,\hat {\upsilon }_{P + 1} )\ .
\end{equation}
\end{theorem}

\vspace{.2in}

\begin{proof}
Let $\theta _{P + 1} = ({\rho}_1,{\omega} _1 ,{\upsilon} _1
,{\varphi}_1,\ldots ,{\rho}_P ,{\omega} _P ,{\upsilon}
_P,{\varphi}_P,{\rho}_{P+1} ,{\omega} _{P+1} ,{\upsilon}
_{P+1},{\varphi}_{P+1} ) $, be some vector in the parameter space
$\Theta _{P + 1} $.
 We have,
\begin{equation}
\begin{array}{l}
 {\cal L}_{P + 1} (\theta _{P + 1} ) = \frac{1}{NM}\sum\limits_{n = 0}^{N-1} \sum\limits_{m = 0}^{M-1}
\bigg( y(n,m) - \sum\limits_{i = 1}^{P+1} \rho_i \cos (\omega _i n
+
\upsilon _i m+\varphi_i)  \bigg)^2 \\
 =  \frac{1}{NM}\sum\limits_{n = 0}^{N-1} \sum\limits_{m = 0}^{M-1}
\bigg( y(n,m) - \sum\limits_{i = 1}^{P} \rho_i \cos (\omega _i n +
\upsilon _i m+\varphi_i)  \bigg)^2 \\
+ \frac{1}{NM}\sum\limits_{n
= 0}^{N - 1} \sum\limits_{m = 0}^{M - 1}\bigg(\rho_{P+1} \cos
(\omega _{P+1} n +
\upsilon _{P+1} m+\varphi_{P+1})  \bigg)^2 \\
 -  \frac{2}{NM}\sum\limits_{n = 0}^{N-1} \sum\limits_{m = 0}^{M-1}
\bigg( y(n,m) - \sum\limits_{i = 1}^{P} \rho_i \cos (\omega _i n +
\upsilon _i m+\varphi_i)  \bigg)\bigg(\rho_{P+1} \cos (\omega
_{P+1} n +
\upsilon _{P+1} m+\varphi_{P+1})  \bigg)\\
 = {\cal L}_P (\theta _P ) + \frac{\rho_{P+1}^2}{2}+\frac{1}{2NM}\sum\limits_{n
= 0}^{N - 1} \sum\limits_{m = 0}^{M - 1}\rho_{P+1}^2 \cos (2\omega
_{P+1} n + 2\upsilon _{P+1} m+2\varphi_{P+1})\\
- \frac{2}{NM}\sum\limits_{n = 0}^{N-1} \sum\limits_{m = 0}^{M-1}
w(n,m)\rho_{P+1} \cos (\omega _{P+1} n +
\upsilon _{P+1} m+\varphi_{P+1})   \\
 - \frac{2}{NM}\sum\limits_{n = 0}^{N-1} \sum\limits_{m = 0}^{M-1}
\bigg( \sum\limits_{i = 1}^{P} \rho_i^0 \cos (\omega _i^0 n +
\upsilon _i^0 m+\varphi_i^0) - \sum\limits_{i = 1}^{P} \rho_i \cos
(\omega _i n + \upsilon _i m+\varphi_i)  \bigg) \\
\bigg(\rho_{P+1} \cos (\omega _{P+1} n + \upsilon _{P+1}
m+\varphi_{P+1})  \bigg)
 = H_1(\theta _{P + 1} ) + H_2(\theta _{P + 1} ) + H_3(\theta _{P + 1} ) \label{lp1}
 \end{array}
\end{equation}
where, $\theta _P = ({\rho}_1,{\omega} _1 ,{\upsilon} _1
,{\varphi}_1,\ldots ,{\rho}_P ,{\omega} _P ,{\upsilon}
_P,{\varphi}_P ) \in \Theta _P $ and,
\begin{equation}
H_1(\theta _{P + 1} ) = {\cal L}_P ({\rho}_1,{\omega} _1
,{\upsilon} _1 ,{\varphi}_1,\ldots ,{\rho}_P ,{\omega} _P
,{\upsilon} _P,{\varphi}_P ) = {\cal L}_P (\theta _P ),
\end{equation}
\begin{equation}
H_2(\theta _{P + 1} ) = \frac{\rho_{P+1}^2}{2} -
\frac{2}{NM}\sum\limits_{n = 0}^{N-1} \sum\limits_{m = 0}^{M-1}
w(n,m)\rho_{P+1} \cos (\omega _{P+1} n + \upsilon _{P+1}
m+\varphi_{P+1}),
\end{equation}
\begin{eqnarray}
&&H_3(\theta _{P + 1} ) =\frac{1}{2NM}\sum\limits_{n = 0}^{N - 1}
\sum\limits_{m = 0}^{M - 1}\rho_{P+1}^2 \cos (2\omega
_{P+1} n + 2\upsilon _{P+1} m+2\varphi_{P+1}) \nonumber\\
 &&- \frac{2}{NM}\sum\limits_{n = 0}^{N-1} \sum\limits_{m = 0}^{M-1}
\bigg( \sum\limits_{i = 1}^{P} \rho_i^0 \cos (\omega _i^0 n +
\upsilon _i^0 m+\varphi_i^0) - \sum\limits_{i = 1}^{P} \rho_i \cos
(\omega _i n + \upsilon _i m+\varphi_i)  \bigg) \nonumber\\
&&\bigg(\rho_{P+1} \cos (\omega _{P+1} n + \upsilon _{P+1}
m+\varphi_{P+1})  \bigg). \label{h3}
\end{eqnarray}

Let $\hat {\theta }_P = (\hat{\rho}_1,\hat{\omega} _1
,\hat{\upsilon} _1 ,\hat{\varphi}_1,\ldots ,\hat{\rho}_P
,\hat{\omega} _P ,\hat{\upsilon} _P,\hat{\varphi}_P )$ be a vector
in $\Theta _P $ that minimizes $H_1(\theta _{P + 1} ) = {\cal L}_P
(\theta _P )$. From \cite{kundu} (or using Theorem 1 in the
previous section),
\begin{equation}
\hat {\theta }_P \to \theta _P^0 \ \ \ a.s. \ \ \ {\rm as} \ \ \
\Psi (N,M) \to \infty. \label{kkkk}
\end{equation}

The function $H_2$  is a function of ${\rho}_{P+1} ,{\omega}
_{P+1} ,{\upsilon} _{P+1},{\varphi}_{P+1}$ only. Evaluating the
partial derivatives of $H_2$ with respect to these variables, it
is easy to verify that the extremum points of $H_2$ are also the
extremum points of the periodogram of the realization of the noise
field. Moreover, let ${\rho}^e ,{\omega} ^e ,{\upsilon}
^e,{\varphi}^e$ denote an extremum point of $H_2$. Then at this
point
\begin{equation}
H_2({\rho}^e ,{\omega} ^e ,{\upsilon}
^e,{\varphi}^e)=-\frac{I_w(\omega^e,\upsilon^e)}{NM}.
\end{equation}

Hence, the minimal value of $H_2$ is obtained at the coordinates
${\rho}_{P+1} ,{\omega} _{P+1} ,{\upsilon} _{P+1},{\varphi}_{P+1}$
where the periodogram of $\{w(n,m)\}$ is maximal. Let
$\hat{\rho}_{P+1} ,\hat{\omega} _{P+1} ,\hat{\upsilon}
_{P+1},\hat{\varphi}_{P+1}$ denote the coordinates that minimize
$H_2 $. Then we have
\begin{equation}
(\hat {\omega }_{P + 1} ,\hat {\upsilon }_{P + 1} ) = \mathop {\arg
\min } \limits_{(\omega ,\upsilon ) \in {(0,2\pi)^2}} H_2
({\rho}_{P+1} ,{\omega} _{P+1} ,{\upsilon} _{P+1},{\varphi}_{P+1}) =
\mathop {\arg \max }\limits_{(\omega ,\upsilon ) \in {(0,2\pi)^2}}
I_w (\omega ,\upsilon ),
\end{equation}
and
\begin{equation}
\hat {\rho}_{P + 1}^2 = \frac{2}{NM}I_w (\hat{\omega}_{P+1}
,\hat{\upsilon}_{P+1} ).
\end{equation}

By Assumption 1, 2 and Theorem 1, \cite{zhang},   we have
\begin{equation}
\mathop{\rm \sup}\limits_
{\omega,\upsilon}I_w(\omega,\upsilon)=O(\log NM). \label{con15}
\end{equation}

Therefore,
\begin{equation}
 H_2 ( \hat{\rho}_{P+1} ,\hat{\omega}
_{P+1} ,\hat{\upsilon}
_{P+1},\hat{\varphi}_{P+1})=O\left(\frac{\log NM}{NM} \right) \ .
\label{h222}
\end{equation}

 Let $\hat {\theta }_{P + 1} \in \Theta _{P + 1} $ be the
vector composed of the elements of the vector $\hat {\theta }_P
\in \Theta _P $ and of  $\hat{\rho}_{P+1} ,\hat{\omega} _{P+1}
,\hat{\upsilon} _{P+1},\hat{\varphi}_{P+1} $, defined  above, \ie,
$$\hat {\theta }_{P + 1} = (\hat{\rho}_1,\hat{\omega} _1
,\hat{\upsilon} _1 ,\hat{\varphi}_1,\ldots ,\hat{\rho}_P
,\hat{\omega} _P ,\hat{\upsilon}
_P,\hat{\varphi}_P,\hat{\rho}_{P+1} ,\hat{\omega} _{P+1}
,\hat{\upsilon} _{P+1},\hat{\varphi}_{P+1} ).$$ We need to verify
that this vector minimizes ${\cal L}_{P + 1} (\theta _{P + 1} )$
on $\Theta _{P + 1} $ as $\Psi (N,M) \to \infty $ .

Recall that for   $\omega \in (0,2\pi)$ and $\varphi\in [0,2\pi)$
\begin{equation}
\sum\limits_{n = 0}^{N - 1} {\cos (\omega n+\varphi)}=\frac{\sin
\left([N -\frac{1}{2}]\omega +\varphi\right)+\sin \left(
\frac{\omega}{2}-\varphi \right)}{2\sin \left( \frac{\omega}{2}
\right)}=O(1). \label{l4x}
\end{equation}
Hence, as $N\to \infty $
\begin{equation}
\frac{1}{ \log N}\sum\limits_{n = 0}^{N - 1} {\cos (\omega
n+\varphi)}=o(1)\ , \label{l4b}
\end{equation}
and consequently
\begin{equation}
\frac{1}{N}\sum\limits_{n = 0}^{N - 1} {\cos (\omega
n+\varphi)}=o\bigg(\frac{\log N}{N}\bigg)\ . \label{l4}
\end{equation}

Next, we evaluate $H_3$.  Consider the first term in (\ref{h3}).
By (\ref{l4}) we have
\begin{equation}
\frac{1}{2NM}\sum\limits_{n = 0}^{N - 1} \sum\limits_{m = 0}^{M -
1}\rho_{P+1}^2 \cos (2\omega _{P+1} n + 2\upsilon _{P+1}
m+2\varphi_{P+1})=o\left(\frac{\log NM}{NM} \right),
\end{equation}
for {\it any} set of values $ {\rho}_{P+1} ,{\omega} _{P+1}
,{\upsilon} _{P+1},{\varphi}_{P+1} $ may assume.

 Consider the second term in (\ref{h3}).
By (\ref{l4}) and unless there exists some $i$, $1 \le i \le P$,
such that $(\omega _{P + 1} , \upsilon _{P + 1})=(\omega_i^0,
\upsilon_i^0)$, we have as $\Psi (N,M) \to \infty$,
\begin{equation} \frac{1}{NM}\sum\limits_{n = 0}^{N -
1} \sum\limits_{m = 0}^{M - 1} \sum\limits_{i = 1}^P\rho_{i}^0
\rho_{P+1}\cos (\omega _i^0 n + \upsilon _i^0 m+\varphi_i^0) \cos
(\omega _{P+1} n + \upsilon _{P+1} m+\varphi_{P+1})
=o\left(\frac{\log NM}{NM} \right), \label{new2}
\end{equation}
for {\it any} set of values $ {\rho}_{P+1} ,{\omega} _{P+1}
,{\upsilon} _{P+1},{\varphi}_{P+1} $ may assume.

Assume now that there exists some $i$, $1 \le i \le P$, such that
$(\omega _{P + 1} , \upsilon _{P + 1})=(\omega_i^0,
\upsilon_i^0)$. Since by assumption there are no two different
regressors  with identical spatial frequencies, it follows that
one of the estimated frequencies $(\omega_i, \upsilon_i)$ is due
to noise contribution.
Hence, by interchanging the roles of $(\omega _{P + 1} , \upsilon
_{P + 1})$ and $(\omega_i, \upsilon_i)$, and repeating the above
argument we conclude that this term has the same order as in
(\ref{new2}). Similarly, for the third term in (\ref{h3}): By
(\ref{l4}) and unless there exists some $i$, $1 \le i \le P$, such
that $(\omega _{P + 1} , \upsilon _{P + 1})=(\omega_i,
\upsilon_i)$, we have as $\Psi (N,M) \to \infty$,
\begin{equation}
 \frac{1}{NM}\sum\limits_{n = 0}^{N -
1} \sum\limits_{m = 0}^{M - 1} \sum\limits_{i = 1}^P\rho_{i}
\rho_{P+1}\cos (\omega _i n + \upsilon _i m+\varphi_i) \cos
(\omega _{P+1} n + \upsilon _{P+1} m+\varphi_{P+1})
=o\left(\frac{\log NM}{NM} \right).
\end{equation}
However such $i$ for which $(\omega _{P + 1} , \upsilon _{P +
1})=(\omega_i, \upsilon_i)$ cannot exist, as this amounts to
reducing the number of regressors from $P+1$ to $P$, as two of
them coincide. Hence, for {\it any} ${\theta }_{P + 1} \in {\Theta
}_{P + 1}$ as $\Psi (N,M) \to \infty$
\begin{equation}
H_3 ( {\theta }_{P + 1} ) = o\left(\frac{\log NM}{NM} \right).
\label{t25}
\end{equation}
On the other hand, the strong consistency (\ref{kkkk}) of the LSE
under the correct model order assumption implies that as $\Psi (N,M)
\to \infty$ the minimal value of ${\cal L}_{P} ( {\theta }_{P} )=
\sigma^2\sum_{(r,s) \in D}a^2(r,s)$ a.s., while from (\ref{h222}) we
have for the minimal value of $H_2$ that $H_2 ( \theta_{P + 1}
)=O\left(\frac{\log NM}{NM} \right)$. Hence, the value of $H_3 (
{\theta }_{P + 1} )$ at {\it any} point in $\Theta_{p+1}$ is
negligible even relative to the values ${\cal L}_{P} ( {\theta }_{P}
)$ and $ H_2 ( \theta_{P + 1} )$ assume at their respective {\it
minimum} points. Therefore, evaluating (\ref{lp1}) as $\Psi (N,M)
\to \infty$ we have
\begin{eqnarray}
  {\cal L}_{P + 1} ( {\theta }_{P + 1})& =&
 {\cal L}_{P} ( {\theta }_{P} )+ H_2 ( {\rho}_{P+1} ,{\omega} _{P+1}
,{\upsilon} _{P+1},{\varphi}_{P+1} )+ H_3 ( {\theta }_{P +
1} )\nonumber \\
&=& {\cal L}_{P} ( {\theta }_{P} )+ H_2 ( {\rho}_{P+1} ,{\omega}
_{P+1} ,{\upsilon} _{P+1},{\varphi}_{P+1} ) + o\left(\frac{\log
NM}{NM} \right)
 . \label{lp2}
\end{eqnarray}
Since $ {\cal L}_{P} ( {\theta }_{P} )$ is a function of the
parameter vector ${\theta }_{P}$ and is independent of
${\rho}_{P+1} ,{\omega} _{P+1} ,{\upsilon}
_{P+1},{\varphi}_{P+1}$, while $H_2$ is a function of
${\rho}_{P+1} ,{\omega} _{P+1} ,{\upsilon} _{P+1},{\varphi}_{P+1}$
and is independent of ${\theta }_{P}$, the problem of minimizing $
{\cal L}_{P + 1} ( {\theta }_{P + 1}) $ becomes  {\it separable}
as ${\Psi (N,M) \to \infty } $. Thus minimizing (\ref{lp2}) is
equivalent to separately minimizing $ {\cal L}_{P} ( {\theta
}_{P}) $ and $ H_2 ( {\rho}_{P+1} ,{\omega} _{P+1} ,{\upsilon}
_{P+1},{\varphi}_{P+1})$ as ${\Psi (N,M) \to \infty } $. Using the
foregoing conclusions, the theorem follows.
\end{proof}

%
%

\subsection{Discussion}

In the above theorems, we have considered the problem of { least
squares} estimation of the parameters of 2-D sinusoidal signals
observed in the presence of an additive colored noise field, when
the assumed number of sinusoids is incorrect. In the case where the
number of sinusoidal signals is under-estimated we have established
the almost sure convergence of the least squares estimates to the
parameters of the dominant sinusoids. This result can be intuitively
explained using the basic principles of least squares estimation:
Since the least squares estimate is the set of model parameters that
minimizes the $\ell_2$ norm of the error between the observations
and the assumed model, it follows that in the case where the model
order is under-estimated the minimum error norm is achieved when the
$k$ most dominant sinusoids are correctly estimated. Similarly, in
the case where the number of sinusoidal signals is over-estimated,
the estimated parameter vector obtained by the least squares
estimator contains a $4P$-dimensional sub-vector that converges
almost surely to the correct parameters of the sinusoids, while the
remaining $k-P$ components assumed to exist, are assigned to the
$k-P$ most dominant spectral peaks of the noise power to further
minimize the norm of the estimation error.
%
%

\section{Strong Consistency of a Family of Model Order Selection
Rules} In this section we employ the results derived in the previous
section in order to establish the strong consistency of a new model
order selection rule.

 It is assumed that there are
$Q$ competing models, where $Q$ is finite, $Q>P$, and that each
competing model $k \in Z_Q =\{0,1,2,\dots,Q-1\}$ is equiprobable.
Following the MDL-MAP template, define the statistic

\begin{equation}
\chi_{\xi}(k)=NM\log {\cal L}_k(\hat\theta _k )+ \xi k \log NM,
\end{equation}
where $\xi$ is some finite constant to be specified later, and
${\cal L}_k(\hat \theta _k )$ is the minimal value of the error
variance of the least squares estimator.

The number of 2-D sinusoids is estimated by minimizing
$\chi_{\xi}(k)$ over $k \in Z_Q$, \ie,
\begin{equation}
\hat P=\mathop{\rm arg \min}\limits_ {k \in Z_Q}\bigg
\{\chi_{\xi}(k)\bigg\}. \label{con1}
\end{equation}

Let

\begin{equation}
{\cal A}:=\frac{\sum_{(r,s) \in D}\sum_{(q,t) \in D}
|a(r,s)a(q,t)|}{\sum_{(r,s) \in D} a^2(r,s)}.
\end{equation}


 The objective of the next theorem is to prove the asymptotic
consistency of the model order selection procedure in
(\ref{con1}).

\begin{theorem}
Let  Assumptions 1-4 be satisfied. Let $\hat P$ be given by
(\ref{con1}) with $\xi>14{\cal A}$. Then as $\Psi(N,M)\rightarrow
\infty$
\begin{equation}
\hat P\rightarrow P \ \ \mbox{a.s.}
\end{equation}
\end{theorem}

\begin{proof}

For $k \leq P$,
\begin{eqnarray}
&& \chi_{\xi}(k-1)-\chi_{\xi}(k) \nonumber \\
&&=NM \log {\cal L}_{k-1}(\hat\theta _{k-1} )+ \xi
(k-1)\log NM -NM \log {\cal L}_k(\hat\theta _k )- \xi k\log NM \nonumber \\
&&=NM \log\bigg(\frac{{\cal L}_{k-1}(\hat\theta _{k-1} )}{{\cal
L}_k(\hat\theta _k )}\bigg)
 - \xi \log NM.
\label{con9}
\end{eqnarray}

From Theorem 1  as $\Psi(N,M)\rightarrow \infty$
\begin{equation}
\hat\theta _{k} \rightarrow \theta _{k}^0 \ \ \mbox{a.s.}, \ \
\label{x1}
\end{equation}
and
\begin{equation}
\hat\theta _{k-1} \rightarrow \theta _{k-1}^0 \ \ \mbox{a.s.} \ \
\end{equation}

From the definition of ${\cal L}_{k}(\hat\theta _{k} )$, and (\ref{x1})
\begin{equation}
\begin{array}{l}
 {\cal L}_k(\hat\theta _k ) = \frac{1}{NM}\sum\limits_{n = 0}^{N-1} \sum\limits_{m = 0}^{M-1}
\bigg( y(n,m) - \sum\limits_{i = 1}^k \hat\rho_i \cos (\hat\omega
_i n +
\hat\upsilon _i m+\hat\varphi_i)  \bigg)^2 \\
 = \frac{1}{NM}\sum\limits_{n = 0}^{N-1} \sum\limits_{m = 0}^{M-1}
\bigg( \sum\limits_{i = 1}^P \rho_i^0 \cos (\omega _i^0 n +
\upsilon _i^0 m+\varphi_i^0) + w(n,m) -\sum\limits_{i = 1}^k
\hat\rho_i \cos (\hat\omega _i n +
\hat\upsilon _i m+\hat\varphi_i) \bigg)^2  \\
 \mathop{\rm \longrightarrow}\limits_ {\Psi(N,M)\rightarrow
\infty} \frac{1}{NM}\sum\limits_{n = 0}^{N-1} \sum\limits_{m =
0}^{M-1} \bigg( \sum\limits_{i = k+1}^P \rho_i^0 \cos (\omega _i^0 n
+ \upsilon _i^0 m+\varphi_i^0) + w(n,m) \bigg)^2.
 \end{array}
\end{equation}

 From Lemma 3 in Appendix C we have
that as $\Psi (N,M) \to \infty$
\begin{equation}
\mathop {\sup }\limits_{\omega,\upsilon } \left |
\frac{1}{NM}\sum\limits_{n = 0}^{N-1} \sum\limits_{m=0}^{M-1} w(n,m)
\cos (\omega n + \upsilon m)
  \right | \to 0 \ \ \mbox{a.s.}
\label{lem}
\end{equation}

Hence, from the  Assumption 3, (\ref{l4x}), (\ref{lem}) and the
Strong Law of Large Numbers, we conclude that as
$\Psi(N,M)\rightarrow \infty$
\begin{equation}
{\cal L}_{k}(\hat\theta _{k} )\rightarrow \sigma^2\sum_{(r,s) \in D}
a^2(r,s)+ \sum_{i=k+1}^{P}\frac{(\rho_i^0)^2}{2}\ \ \mbox{a.s.}
\end{equation}
and similarly
\begin{equation}
{\cal L}_{k-1}(\hat\theta _{k-1} )\rightarrow \sigma^2\sum_{(r,s)
\in D} a^2(r,s)+ \sum_{i=k}^{P}\frac{(\rho_i^0)^2}{2}\ \ \mbox{a.s.}
\label{con8}
\end{equation}

Since $\frac{\log NM}{NM}$ tends to zero, as $\Psi(N,M)\rightarrow
\infty$, then as $\Psi(N,M)\rightarrow \infty$
\begin{eqnarray}
(NM)^{-1}(\chi_{\xi}(k-1)-\chi_{\xi}(k)) \rightarrow \log
\bigg(1+\frac{(\rho_k^0)^2}{2\sigma^2\sum_{(r,s) \in D}
 a^2(r,s)+
\sum_{i=k+1}^{P}(\rho_i^0)^2}\bigg)  \mbox{a.s.}
\end{eqnarray}

Since $ \log \bigg(1+\frac{(\rho_k^0)^2}{2\sigma^2\sum_{(r,s) \in D}
 a^2(r,s)+
\sum_{i=k+1}^{P}(\rho_i^0)^2}\bigg)$ is strictly positive, then
$\chi_{\xi}(k-1)> \chi_{\xi}(k)$. Hence, for $k \leq P$, the
function $\chi_{\xi}(k)$ is monotonically decreasing with $k$.

We next consider  the case where $k=P+l$ for any integer $l\geq1$.

Based on \cite{zhang}, Theorem 1 and Assumptions 1, 2 we have that
\begin{equation}
\mathop{\rm \lim \sup}\limits_ {\Psi(N,M)\rightarrow
\infty}\frac{\mathop{\rm \sup}\limits_
{\omega,\upsilon}I_{w}(\omega,\upsilon)}{ {\mathop{\rm
\sup}\limits_ {\omega,\upsilon}f_w(\omega,\upsilon)}\log (NM)}\leq
14 \  \ \ \mbox{a.s.} \label{con15a}
\end{equation}

Based on an extension of Theorem 2 we have that a.s. as
$\Psi(N,M)\rightarrow \infty$
\begin{equation}
{\cal L}_{P+l}(\hat\theta _{P+l} )={\cal L}_P(\hat\theta _P )-
\frac{U_l}{NM}+o\bigg(\frac{\log NM}{NM}\bigg), \label{con16}
\end{equation}
where
\begin{equation}
U_l=\sum_{i=1}^lI_w(\omega_i,\upsilon_i),
\end{equation}
is the sum of the $l$ largest elements of the periodogram of the
noise field $\{w(s,t)\}$. Clearly
\begin{equation}
U_l\leq l\mathop{\rm \sup}\limits_
{\omega,\upsilon}I_u(\omega,\upsilon).
\end{equation}


Similarly to (\ref{con9}), a.s. as $\Psi(N,M)\rightarrow \infty$,
\begin{eqnarray}
&&\chi_{\xi}(P+l)-\chi_{\xi}(P) \nonumber \\
&&=NM \log {\cal L}_{P+l}(\hat\theta _{P+l} )+ \xi (P+l)\log NM -
NM \log {\cal L}_P(\hat\theta _P )- \xi P\log NM \nonumber \\
&&=\xi l\log NM+NM\log \left(1-\frac{U_l}{NM {\cal L}_P(\hat\theta
_P )}+o\bigg(\frac{\log NM}{NM}\bigg)\right) \nonumber \\
&&=\xi l\log NM-\bigg(\frac{U_l}{ {\cal L}_P(\hat\theta _P )} +o(\log NM)\bigg)(1+o(1))\nonumber \\
&&=\log NM\bigg(\xi l - \frac{U_l}{{\cal L}_P(\hat\theta _P )\log
NM}+o(1)\bigg)
\geq\log NM\bigg(\xi l - \frac{l\mathop{\rm \sup}\limits_ {\omega,\upsilon}I_w(\omega,\upsilon)}{{\cal L}_P(\hat\theta _P )\log NM}+o(1)\bigg) \nonumber \\
&&= l\log NM\bigg(\xi  - \frac{\mathop{\rm \sup}\limits_
{\omega,\upsilon}I_w(\omega,\upsilon)}{\mathop{\rm \sup}\limits_
{\omega,\upsilon}f_w(\omega,\upsilon)\log NM}\frac{\mathop{\rm
\sup}\limits_ {\omega,\upsilon}f_w(\omega,\upsilon)}{{\cal
L}_P(\hat\theta _P )}+o(1)\bigg), \label{con14}
\end{eqnarray}
where the second equality is obtained by substituting ${\cal
L}_{P+l}(\hat\theta _{P+l} )$  using the equality (\ref{con16}).
The third equality is due to the property that for $x\rightarrow
0$, $\log(1+x)=x(1+o(1))$, where the observation that the term
$\frac{U_l}{NM{\cal L}_P(\hat\theta _P )}$ tends to zero a.s. as
$\Psi(N,M)\rightarrow \infty$ is due to (\ref{con15a}).

From \cite{kundu} (or using Theorem 1 in the previous section),
\begin{equation}
\hat {\theta }_P \to \theta _P^0 \ \ \ a.s. \ \ \ {\rm as} \ \ \
\Psi (N,M) \to \infty. \label{kkkk1}
\end{equation}

Hence, the strong consistency (\ref{kkkk1}) of the LSE under the
correct model order assumption implies that as $\Psi (N,M) \to
\infty$
\begin{equation}
{\cal L}_{P}(\hat\theta _{P} )\rightarrow \sigma^2\sum_{(r,s) \in D}
a^2(r,s)\ \ \mbox{a.s.} \label{con16a}
\end{equation}
On the other hand using the triangle inequality
\begin{equation}
\mathop{\rm \sup}\limits_{\omega,\upsilon}f_w(\omega,\upsilon)\leq
\sigma^2\sum_{(r,s) \in D}\sum_{(q,t)\in D} |a(r,s)a(q,t)|.
\label{con17}
\end{equation}

 Substituting (\ref{con15a}),(\ref{con16a}) and (\ref{con17}) into (\ref{con14})
we conclude that
\begin{equation}
\chi_{\xi}(P+l)-\chi_{\xi}(P)>0
\end{equation}
for any integer $l\geq1$. Therefore, a.s. as $\Psi(N,M)\rightarrow
\infty$, the function $\chi_{\xi}(k)$ has a {\bf global minimum} for
$k=P$.
\end{proof}

\section{Special Case}

Introducing some additional restrictions on the structure of the
noise field, we can establish a tighter (in terms of $\xi$) model
order selection rule. We thus modify our earlier Assumption 1, 2
regarding the noise field as follows:

\textbf{ Assumption 1'} The noise field $\{w(n,m)\}$
 is an infinite order \emph{quarter-plane} MA field, \ie,
\begin{equation}
w(n,m) = \sum_{r,s=0}^{\infty}a(r,s) u(n-r,m-s)  \label{e3a}
\end{equation}
where the field $\{u(n,m)\}$ is an i.i.d. real valued zero-mean
random field with finite variance $\sigma^2$, such that ${
E}[u(n,m)^2 \log|u(n,m)|]< \infty $.

\textbf{ Assumption 2'} The sequence ${a(i,j)}$ is a deterministic
sequence which satisfied the condition
\begin{equation}
\sum_{r,s=0}^{\infty}(r+s)|a(r,s)| < \infty.
\end{equation}

In this case, based on \cite{He}, Theorem 3.2 and Assumption 1', 2'
we have that
\begin{equation}
\mathop{\rm \lim \sup}\limits_ {\Psi(N,M)\rightarrow
\infty}\frac{\mathop{\rm \sup}\limits_
{\omega,\upsilon}I_{w}(\omega,\upsilon)}{ {\mathop{\rm \sup}\limits_
{\omega,\upsilon}f_w(\omega,\upsilon)}\log (NM)}\leq 8 \  \ \
\mbox{a.s.} \label{con15}
\end{equation}

The results of Theorem 1 and 2 are not affected by this assumption.
The only change is in Theorem 3. Therefore we can formulate the next
theorem:

\begin{theorem}
Let  Assumptions 1', 2', 3 and 4 be satisfied. Let $\hat P$ be given
by (\ref{con1}) with $\xi>8{\cal A}$. Then as $\Psi(N,M)\rightarrow
\infty$
\begin{equation}
\hat P\rightarrow P \ \ \mbox{a.s.}
\end{equation}
\end{theorem}

The proof of this theorem is identical to the proof of Theorem 3,
where instead of (\ref{con15a}) we employ the inequality in
(\ref{con15}).

%
%
\section{Conclusions}
\label{Concl} We have considered the problem of jointly estimating
the number as well as the  parameters of two-dimensional sinusoidal
signals, observed in the presence of an additive colored noise
field. We have established the strong consistency of the LSE when
the number of sinusoidal signals is under-estimated, or
over-estimated. Based on these results, we have proved the strong
consistency of a new model order selection rule for the number of
sinusoidal components.

\vspace{0.3in}

%
%
\appendix{\large {\bf Appendix A}}
\setcounter{lemma}{0}
\begin{lemma}
\begin{equation}
\mathop {\liminf}\limits_{\Psi (N,M) \to \infty } \mathop {\inf
}\limits_{\theta _k \in \Delta _\delta } \left( {{\cal L}_k(\theta
_k ) - {\cal L}_k(\theta _k^0 )} \right) > 0 \ \ a.s. \label{l0lm}
\end{equation}
\end{lemma}
\begin{proof}

In the following we first show that on $\Delta _\delta$ the
sequence ${\cal L}_k(\theta _k ) - {\cal L}_k(\theta _k^0 )$
(indexed in $N,M$) is uniformly lower bounded by a strictly
positive constant as $\Psi (N,M) \to \infty$. Since the sequence
elements are uniformly lower bounded by a strictly positive
constant the sequence of infimums, $\mathop {\inf }\limits_{\theta
_k \in \Delta _\delta } \left( {{\cal L}_k(\theta _k ) - {\cal
L}_k(\theta _k^0 )} \right)$, is  uniformly lower bounded by the
same strictly positive constant as $\Psi (N,M) \to \infty$, and
hence, $\mathop {\liminf}\limits_{\Psi (N,M) \to \infty } \mathop
{\inf }\limits_{\theta _k \in \Delta _\delta } \left( {{\cal
L}_k(\theta _k ) - {\cal L}_k(\theta _k^0 )} \right)$.

Thus, we first prove that the sequence ${\cal L}_k(\theta _k ) -
{\cal L}_k(\theta _k^0 )$  is uniformly lower bounded away from
zero on $\Delta _\delta$ as $\Psi (N,M) \to \infty$.
\begin{equation}
\begin{array}{l}
 {\cal L}_k(\theta _k ) - {\cal L}_k(\theta _k^0 ) \\
 = \frac{1}{NM}\sum\limits_{n = 0}^{N-1} \sum\limits_{m = 0}^{M-1}
\bigg( y(n,m) - \sum\limits_{i = 1}^k \rho_i \cos (\omega _i n +
\upsilon _i m+\varphi_i)  \bigg)^2 \\
-\frac{1}{NM}\sum\limits_{n = 0}^{N-1} \sum\limits_{m = 0}^{M-1}
\bigg( y(n,m) - \sum\limits_{i = 1}^k \rho_i^0 \cos (\omega _i^0 n
+
\upsilon _i^0 m+\varphi_i^0)  \bigg)^2 \\
 = \frac{1}{NM}\sum\limits_{n = 0}^{N-1} \sum\limits_{m = 0}^{M-1}
\bigg( \sum\limits_{i = 1}^P \rho_i^0 \cos (\omega _i^0 n +
\upsilon _i^0 m+\varphi_i^0) + w(n,m) -\sum\limits_{i = 1}^k
\rho_i \cos (\omega _i n +
\upsilon _i m+\varphi_i) \bigg)^2  \\
 - \frac{1}{NM}\sum\limits_{n = 0}^{N-1} \sum\limits_{m = 0}^{M-1}
\bigg( \sum\limits_{i = k+1}^P \rho_i^0 \cos (\omega _i^0 n +
\upsilon _i^0 m+\varphi_i^0) + w(n,m) \bigg)^2 \\
 = \frac{1}{NM}\sum\limits_{n = 0}^{N-1} \sum\limits_{m = 0}^{M-1}
\bigg( \sum\limits_{i = 1}^k \rho_i^0 \cos (\omega _i^0 n +
\upsilon _i^0 m+\varphi_i^0)  - \sum\limits_{i = 1}^k \rho_i \cos
(\omega _i n +
\upsilon _i m+\varphi_i)  \bigg)^2 \\
 +  \frac{2}{NM}\sum\limits_{n = 0}^{N-1} \sum\limits_{m=0}^{M-1} \bigg( \sum\limits_{i = k+1}^P \rho_i^0 \cos (\omega _i^0 n +
\upsilon _i^0 m+\varphi_i^0) \bigg)\\
 \bigg( \sum\limits_{i =
1}^k \rho_i^0 \cos (\omega _i^0 n + \upsilon _i^0 m+\varphi_i^0) -
\sum\limits_{i = 1}^k \rho_i \cos (\omega _i n +
\upsilon _i m+\varphi_i)  \bigg)   \\
 +  \frac{2}{NM}\sum\limits_{n = 0}^{N-1} \sum\limits_{m=0}^{M-1} w(n,m) \bigg( \sum\limits_{i = 1}^k \rho_i^0 \cos (\omega _i^0 n +
\upsilon _i^0 m+\varphi_i^0)
 -\sum\limits_{i = 1}^k \rho_i \cos (\omega _i n +
\upsilon _i m+\varphi_i)  \bigg) \\
 = I_1 + I_2 + I_3. \\
 \end{array}
\label{l1}
\end{equation}
Thus, to check the asymptotic behavior of L.H.S. of (\ref{l1}) we
have to evaluate $ \mathop {\lim}\limits_{\Psi (N,M) \to \infty }
(I_1 + I_2 + I_3 )$ for all vectors  $\theta _k \in \Delta
_\delta$:
\begin{equation}
\begin{array}{l}
 \mathop {\lim }\limits_{\Psi (N,M) \to \infty }  I_1 =  \mathop {\lim
}\limits_{\Psi (N,M) \to \infty }  \frac{1}{NM}\sum\limits_{n =
0}^{N-1} \sum\limits_{m = 0}^{M-1} \left( \sum\limits_{i = 1}^k
\rho_i^0
\cos (\omega _i^0 n + \upsilon _i^0 m+\varphi_i^0)  \right)^2 \\
- \mathop {\lim}\limits_{\Psi (N,M) \to \infty }  \left[ 2
\frac{1}{NM}\sum\limits_{n = 0}^{N-1} \sum\limits_{m = 0}^{M-1}
\sum\limits_{i = 1}^k \sum\limits_{j = 1}^k\rho_i\rho_j^0 \cos
(\omega _i n + \upsilon _i m+\varphi_i)\cos (\omega _j^0 n +
\upsilon _j^0 m+\varphi_j^0) \right]\\
+ \mathop {\lim }\limits_{\Psi (N,M) \to \infty }
\frac{1}{NM}\sum\limits_{n = 0}^{N-1} \sum\limits_{m = 0}^{M-1}
\left( \sum\limits_{i = 1}^k \rho_i \cos (\omega _i n + \upsilon
_i m+\varphi_i) \right)^2
  = T_1 + T_2 + T_3.
 \end{array}
\label{l3a}
\end{equation}

Recall that for   $\vert \rho\vert < \infty $ and $\varphi\in
[0,2\pi)$
\begin{equation}
\mathop {\lim }\limits_{N \to \infty } \frac{1}{N}\sum\limits_{n=
0}^{N - 1} {\rho\cos (\omega n+\varphi)} = 0, \label{l4old}
\end{equation}
uniformly in $\omega$ on any closed interval in $(0,2\pi)$. The same
equality is hold for the sine function. Hence, due to Assumption 3
and (\ref{l4old}), we have
\begin{eqnarray}
 T_1 = \mathop {\lim}\limits_{\Psi (N,M) \to \infty }
\frac{1}{NM}\sum\limits_{n = 0}^{N-1} \sum\limits_{m =
0}^{M-1}\bigg( \sum\limits_{i = 1}^k \rho_i^0 \cos (\omega _i^0 n +
\upsilon _i^0 m+\varphi_i^0)  \bigg)^2 =\sum\limits_{i = 1}^k
\frac{(\rho_i^0 ) ^2}{2}, \label{14a}
\end{eqnarray}
independently of $\theta _k$.

 Also,
\begin{equation}
\begin{array}{l}
 T_3 = \mathop {\lim }\limits_{\Psi (N,M)
\to \infty } \frac{1}{NM}\sum\limits_{n = 0}^{N-1} \sum\limits_{m
= 0}^{M-1} \left( \sum\limits_{i = 1}^k \rho_i \cos (\omega _i n +
\upsilon
_i m+\varphi_i) \right)^2 =   \sum\limits_{i = 1}^k\frac{(\rho_i ) ^2}{2} \\
+ \mathop {\lim }\limits_{\Psi (N,M) \to \infty }
\frac{1}{NM}\sum\limits_{n = 0}^{N-1} \sum\limits_{m = 0}^{M-1}
\mathop {\sum\limits_{i = 1 }^k}\limits_{i\ne j}\sum\limits_{j =
1}^k \rho_i \rho_j \cos (\omega _i n + \upsilon _i m+\varphi_i)\cos
(\omega _j n + \upsilon
_j m+\varphi_j).  \\
 \end{array}
\end{equation}
Since the pairs $(\omega_i,\upsilon_i) $ are pairwise different,
then on any closed interval in $(0,2\pi)$ the sequence of partial
sums  $\frac{1}{NM}\sum\limits_{n = 0}^{N-1} \sum\limits_{m =
0}^{M-1} \mathop {\sum\limits_{i = 1 }^k}\limits_{i\ne
j}\sum\limits_{j = 1}^k \rho_i \rho_j \cos (\omega _i n + \upsilon
_i m+\varphi_i)\cos (\omega _j n + \upsilon _j m+\varphi_j)$
converges uniformly to zero as $\Psi (N,M) \to \infty$.

Hence,
\begin{equation}
\begin{array}{l}
 T_3 = \sum\limits_{i = 1}^k
\frac{(\rho_i ) ^2}{2}, \label{xxx}
\end{array}
\end{equation}
 as $\Psi (N,M) \to \infty$ uniformly on $\Delta _\delta$.

 Leaving
$T_2$ unchanged we obtain
\begin{equation}
\begin{array}{l}
 \mathop {\lim }\limits_{\Psi (N,M) \to \infty }  I_1 =  \sum\limits_{i = 1}^k \bigg(\frac{(\rho_i^0 ) ^2}{2}+\frac{(\rho_i ) ^2}{2}\bigg)  \\
- \mathop {\lim }\limits_{\Psi (N,M) \to \infty }
\frac{2}{NM}\sum\limits_{n = 0}^{N-1} \sum\limits_{m = 0}^{M-1}
\sum\limits_{i = 1}^k \sum\limits_{j = 1}^k\rho_i\rho_j^0 \cos
(\omega _i n + \upsilon _i m+\varphi_i)\cos (\omega _j^0 n +
\upsilon _j^0 m+\varphi_j^0),
 \end{array}
\label{l3}
\end{equation}
uniformly on $\Delta _\delta$.

Using the similar considerations to those employed in the
evaluation of (\ref{14a}) we obtain
 \begin{equation}
\begin{array}{l}
 \mathop {\lim }\limits_{\Psi (N,M) \to \infty }  I_2 = \mathop {\lim
}\limits_{\Psi (N,M) \to \infty }
\bigg[\frac{2}{NM}\sum\limits_{n = 0}^{N-1}
\sum\limits_{m=0}^{M-1} \bigg( \sum\limits_{i = k+1}^P \rho_i^0
\cos (\omega _i^0 n +
\upsilon _i^0 m+\varphi_i^0) \bigg)\\
 \bigg( \sum\limits_{i =
1}^k \rho_i^0 \cos (\omega _i^0 n + \upsilon _i^0 m+\varphi_i^0) -
\sum\limits_{i = 1}^k \rho_i \cos (\omega _i n +
\upsilon _i m+\varphi_i)  \bigg)\bigg]    \\
 = - \mathop {\lim}\limits_{\Psi (N,M) \to \infty }  \left[
\frac{2}{NM}\sum\limits_{n = 0}^{N-1} \sum\limits_{m = 0}^{M-1}
\sum\limits_{i = 1}^k \sum\limits_{j = k+1}^P\rho_i\rho_j^0 \cos
(\omega _i n + \upsilon _i m+\varphi_i)\cos (\omega _j^0 n +
\upsilon _j^0 m+\varphi_j^0) \right].
 \end{array}
\label{l5}
\end{equation}
By Lemma 3 in Appendix C, we have that a.s. as $ \Psi (N,M) \to
\infty $ :
\begin{equation}
\mathop {\sup }\limits_{\theta _k \in \Delta _\delta } \left |
\frac{2}{NM}\sum\limits_{n = 0}^{N-1} \sum\limits_{m=0}^{M-1}
w(n,m) \bigg( \sum\limits_{i = 1}^k \rho_i^0 \cos (\omega _i^0 n +
\upsilon _i^0 m+\varphi_i^0)
 -\sum\limits_{i = 1}^k \rho_i \cos (\omega _i n +
\upsilon _i m+\varphi_i)  \bigg) \right | \to 0. \label{l6}
\end{equation}
Hence $I_3 \to 0\enskip a.s.$  as $ \Psi (N,M) \to \infty $
uniformly on $\Delta _\delta$. Using (\ref{l3}), (\ref{l5}) and
(\ref{l6}) we conclude that a.s.
\begin{equation} \begin{array}{l}
 \mathop {\lim}\limits_{\Psi (N,M) \to \infty }  \left( {{\cal L}_k(\theta
_k ) - {\cal L}_k(\theta _k^0 )} \right)=  \sum\limits_{i = 1}^k \bigg(\frac{(\rho_i^0 ) ^2}{2}+\frac{(\rho_i ) ^2}{2}\bigg)  \\
- \mathop {\lim}\limits_{\Psi (N,M) \to \infty }
\frac{2}{NM}\sum\limits_{n = 0}^{N-1} \sum\limits_{m = 0}^{M-1}
\sum\limits_{i = 1}^k \sum\limits_{j = 1}^P\rho_i\rho_j^0 \cos
(\omega _i n + \upsilon _i m+\varphi_i)\cos (\omega _j^0 n +
\upsilon _j^0 m+\varphi_j^0).
 \end{array}
\label{l7}
\end{equation}
To complete the evaluation of (\ref{l7}) we consider the vectors
$\theta_k \in \Delta _\delta$. Let us first assume that
$\Delta_\delta \equiv {\cal R}_{q \delta} $ for some $q$, $1 \le q
\le k$. Thus, the coordinate $\rho_q$ of each vector in this subset
is different from the corresponding  coordinate $\rho_q^0$ by at
least $\delta>0$. Consider first the case where all the other
elements of the vector $\theta_k \in {\cal R}_{q \delta }$ are
identical to the corresponding elements of $\theta _k^0 $. Since by
this assumption $\omega_j=\omega_j^0$, $\upsilon_j=\upsilon_j^0$,
$\varphi_j=\varphi_j^0$ for $1 \le j \le k$, and $\rho_j=\rho_j^0$
for $1 \le j \le k$, $j \neq q$, on this set we have
\begin{eqnarray}
&&\mathop {\lim}\limits_{\Psi (N,M) \to \infty }  \left( {{\cal
L}_k(\theta _k ) - {\cal L}_k(\theta _k^0 )} \right) = \bigg(
\frac{\rho_q^0}{\sqrt{2}} - \frac{\rho_q}{\sqrt{2}} \bigg) ^2
\nonumber \\
&&-\mathop {\lim }\limits_{\Psi (N,M) \to \infty }
\frac{2}{NM}\sum\limits_{n = 0}^{N-1} \sum\limits_{m = 0}^{M-1}
\mathop{\sum\limits_{i = 1}^k}\limits_{i\ne j}\sum\limits_{j =
1}^P\rho_i\rho_j^0 \cos (\omega _i n + \upsilon _i
m+\varphi_i)\cos (\omega _j^0 n +
\upsilon _j^0 m+\varphi_j^0)\nonumber \\
&&= \bigg( \frac{\rho_q^0}{\sqrt{2}} - \frac{\rho_q}{\sqrt{2}}
\bigg) ^2\ge \frac{\delta^2}{2} > 0, \label{l8}
\end{eqnarray}
uniformly in $\rho_q$, where the second equality is due to
Assumption 3 and following the arguments employed to obtain
(\ref{xxx}).

Assume next that $\theta_k \in {\cal R}_{ q \delta }$ (\ie, the
coordinate $\rho_q$  is different from the corresponding
coordinate $\rho_q^0$ by at least $\delta>0$) and that in
addition, there exists an element $\rho_t$ of $\theta_k$, such
that $1 \le t \le k$,  $t \ne q$ and $\vert \rho_t - \rho_t^0
\vert \ge \lambda ,\lambda > 0$ while all the other elements of
the vector $\theta_k$ are identical to the corresponding elements
of $\theta _k^0 $. Following a similar derivation to the one in
(\ref{l8}) we conclude that

\begin{equation}
\mathop {\lim}\limits_{\Psi (N,M) \to \infty }  \left( {{\cal
L}_k(\theta _k ) - {\cal L}_k(\theta _k^0 )} \right) = \bigg(
\frac{\rho_q^0}{\sqrt{2}} - \frac{\rho_q}{\sqrt{2}} \bigg)^2 +
\bigg( \frac{\rho_t^0}{\sqrt{2}} - \frac{\rho_t}{\sqrt{2}}
\bigg)^2 \ge \frac{\delta^2}{2}+\frac{\lambda^2}{2}
>\frac{\delta^2}{2}, \label{l8b}
\end{equation}
uniformly in $\rho_q$ and $\rho_t$.

 Consider the case where $\theta_k
\in {\cal R}_{ q \delta }$ while there exists an element
$\varphi_l$ of $\theta_k \in {\cal R}_{ q \delta }$, such that
$\vert \varphi_l - \varphi_l^0 \vert \ge \eta ,\eta
> 0$ and  all the other elements of the vector $\theta_k$ are
identical to the corresponding elements of $\theta _k^0$.
Following a similar derivation to the one in (\ref{l8}) we
conclude that
\begin{equation}
\mathop {\lim}\limits_{\Psi (N,M) \to \infty }  \left( {{\cal
L}_k(\theta _k ) - {\cal L}_k(\theta _k^0 )} \right) =\biggl\{
\begin{array}{*{20}c}
 \bigg(
\frac{\rho_q^0}{\sqrt{2}} - \frac{\rho_q}{\sqrt{2}} \bigg)^2 +(\rho_l^0)^2  -(\rho_l^0)^2\cos(\varphi_l-\varphi_l^0),\quad l \ne q  \\
 \frac{(\rho_q^0)^2}{2}+ \frac{(\rho_q)^2}{2}-\rho_q^0\rho_q\cos(\varphi_q-\varphi_q^0),\quad l = q  \\
\end{array}
 >\frac{\delta^2}{2},
\label{l8c}
\end{equation}
uniformly in $\rho_q$ and $\varphi_l$.

 Finally, consider the case
where $\theta_k \in {\cal R}_{ q \delta }$ while there exists an
element $\omega_l$ of $\theta_k \in {\cal R}_{ q \delta }$, such
that $\vert \omega_l - \omega_l^0 \vert \ge \eta ,\eta > 0$ and
all the other elements of the vector $\theta_k$ are identical to
the corresponding elements of $\theta _k^0$. Following a similar
derivation to the one in (\ref{l8}) we conclude that
\begin{equation}
\mathop {\liminf}\limits_{\Psi (N,M) \to \infty }  \left( {{\cal
L}_k(\theta _k ) - {\cal L}_k(\theta _k^0 )} \right) =\biggl\{
\begin{array}{*{20}c}
 \bigg(
\frac{\rho_q^0}{\sqrt{2}} - \frac{\rho_q}{\sqrt{2}} \bigg)^2 +(\rho_l^0)^2,\quad l \ne q  \\
 \frac{(\rho_q^0)^2}{2}+ \frac{(\rho_q)^2}{2},\quad l = q  \\
\end{array}
 >\frac{\delta^2}{2},
\label{l8c1}
\end{equation}
uniformly in $\rho_q$ and  $\omega_l$.

From the above analysis it is clear that $\mathop
{\lim}\limits_{\Psi (N,M) \to \infty }  \left( {{\cal L}_k(\theta
_k ) - {\cal L}_k(\theta _k^0 )} \right)$ is lower bounded by
$\frac{\delta^2}{2}$ uniformly in ${\cal R}_{q \delta }$.

Following similar reasoning, the next subset we consider is $W_{q
\delta } \cup V_{q \delta }$. We first consider a subset of this
set:
\begin{equation}
\Lambda = \left\{ {\theta _k \in W_{q\delta } \cup V_{q\delta }
:\;\,\exists p,\,k + 1 \le p \le P,\;(\omega _q ,\upsilon _q ) =
(\omega _p^0 ,\upsilon _p^0 )\;} \right\} \subset W_{q\delta }
\cup V_{q\delta } \label{l9}
\end{equation}
This subset includes vectors in $\Theta _k $, such that their
coordinate pairs $(\omega _q ,\upsilon _q )$ are different from
the corresponding pairs of $\theta _k^0 $ and equal to some pair
$(\omega _p^0 ,\upsilon _p^0 )$ where $p \ge k + 1$.
 As above, the minimum is obtained when all the other elements of $\theta _k$ are identical to the corresponding elements
of $\theta _k^0 $. Hence, uniformly on $\Lambda$, we have
\begin{equation}
\begin{array}{l}
 \mathop {\lim}\limits_{\Psi (N,M) \to \infty }  \left( {{\cal L}_k(\theta _k ) -
{\cal L}_k(\theta _k^0 )} \right) \geq   {\frac{(\rho_q^0)^2}{2}+ \frac{(\rho_q)^2}{2} - \rho_p^0\rho_q }  \\
 =  {\frac{(\rho_q^0)^2}{2}-\frac{(\rho_p^0)^2}{2} + \bigg(
\frac{\rho_p^0}{\sqrt{2}} - \frac{\rho_q}{\sqrt{2}} \bigg)^2} \geq \frac{(\rho_q^0)^2}{2}-\frac{(\rho_p^0)^2}{2} = \epsilon_{\Lambda} > 0, \\
 \end{array}
\label{l10}
\end{equation}
where the last inequality is due to Assumption 4.

On the complementary set:
\begin{equation}
\Lambda ^c = \left( {W_{q \delta } \cup V_{q \delta } }
\right)\backslash \Lambda = \left\{ {\theta _k \in W_{q \delta }
\cup V_{q \delta } :\;(\omega _q ,\upsilon _q ) \ne (\omega _p^0
,\upsilon _p^0 ),\,\forall p,\,k + 1 \le p \le P\;} \right\}
\label{l11}
\end{equation}
we have
\begin{equation}
\begin{array}{l}
 \mathop {\lim}\limits_{\Psi (N,M) \to \infty }  \left( {{\cal L}_k(\theta _k ) -
{\cal L}_k(\theta _k^0 )} \right) \geq \frac{(\rho_q^0)^2}{2}+ \frac{(\rho_q)^2}{2}\geq \frac{(\rho_q^0)^2}{2}=\epsilon_{\Lambda^c}> 0. \\
 \\
 \end{array}
\label{l12}
\end{equation}

Finally, on the set $\Phi_{q\delta }$  the coordinate $\varphi_q$
of the each vector in this subset is different from the
corresponding coordinate $\varphi_q^0$ by at least $\delta>0$. As
in previous cases , the minimum is obtained when all the other
elements of $\theta _k \in \Phi_{q\delta }$ are identical to the
corresponding elements of $\theta _k^0 $. Hence, uniformly on
$\Phi_{q\delta }$, we have
\begin{equation}
\begin{array}{l}
 \mathop {\lim}\limits_{\Psi (N,M) \to \infty }  \left( {{\cal L}_k(\theta
_k ) - {\cal L}_k(\theta _k^0 )} \right) \geq (\rho_q^0)^2 - (\rho_q^0)^2 \cos(\varphi_q-\varphi_q^0)\geq (\rho_q^0)^2(1-\cos \delta) = \epsilon_{\Phi_{q\delta }}> 0. \\
 \end{array}
\label{l13}
\end{equation}

Let
$\epsilon_q=\min(\frac{\delta^2}{2},\epsilon_{\Lambda},\epsilon_{\Lambda^c},\epsilon_{\Phi_{q\delta
}})$. Collecting (\ref{l8}),(\ref{l10}), (\ref{l12}) and
(\ref{l13}) together we conclude that the sequence ${\cal
L}_k(\theta _k ) - {\cal L}_k(\theta _k^0 )$ is lower bounded by
$\epsilon_q >0$ uniformly on ${\cal R}_{ q \delta } \cup
\Phi_{q\delta } \cup W_{q \delta } \cup V_{q \delta }$ as $\Psi
(N,M) \to \infty$.

By repeating the same arguments for every $q$, $1\leq q \leq k$,
and
 by letting $\epsilon=\min(\epsilon_1,\dots,\epsilon_k)$, we conclude
that the sequence ${\cal L}_k(\theta _k ) - {\cal L}_k(\theta _k^0
)$ (indexed in $N,M$) is lower bounded by $\epsilon >0$ uniformly
on $\Delta_\delta$ as $\Psi (N,M) \to \infty$.

Hence, it follows that sequence $\mathop {\inf }\limits_{\theta _k
\in \Delta _\delta } ({\cal L}_k(\theta _k ) - {\cal L}_k(\theta
_k^0 ))$ (indexed in $N,M$)  is also asymptotically lower bounded
by $\epsilon
>0$, \ie,
\begin{equation}
 \mathop {\inf
}\limits_{\theta _k \in \Delta _\delta } \left( {{\cal L}_k(\theta
_k ) - {\cal L}_k(\theta _k^0 )} \right) \geq \epsilon,
\end{equation}
as $\Psi (N,M) \to \infty$.

Hence, by the definition of $\mathop {\liminf}$
\begin{equation}
\mathop {\liminf}\limits_{\Psi (N,M) \to \infty } \mathop {\inf
}\limits_{\theta _k \in \Delta _\delta } \left( {{\cal L}_k(\theta
_k ) - {\cal L}_k(\theta _k^0 )} \right) \geq \epsilon  > 0.
\label{l8a}
\end{equation}

\end{proof}

\vspace{0.3in}

%
%
\appendix{\large {\bf Appendix B}}
\begin{lemma}
\label{Wulm} Let $\{ x_n,n \ge 1 \}$ be a sequence of random
variables. Then
\begin{equation}
\Pr \{x_n \le 0\ i.o. \} \le \Pr\{\mathop {\liminf }\limits_{n \to
\infty} x_n \le 0\}
\end{equation}
\end{lemma}
\begin{proof}
Let $(\Omega,\Sigma,p)$ be some probability space. Let $\{
x_n(\omega),n \ge 1 \}$ be a sequence of random variables. Let
$\{A_n \in \Sigma,n \ge 1 \}$ be a sequence of subsets of
$\Omega$, such that $A_n = \{ \omega \in \Omega : x_n( \omega )
\le  0 \}$. Define
\begin{equation}
A^m_n={\bigcup\limits_{n = m}^\infty {\{\omega :x_n \le 0 \}} } \
.
\end{equation}
Then
\begin{equation}
A^m_n \subseteq
 \{\omega :\mathop {\inf }\limits_{n \ge m} x_n \le
0 \} \ .
\end{equation}
Hence
\begin{equation}
\bigcap\limits_{m 1}^\infty A^m_n \subseteq \bigcap\limits_{m
1}^\infty  \{\omega :\mathop {\inf }\limits_{n \ge m} x_n \le 0 \}
\ . \label{el3}
\end{equation}
Consider the R.H.S. of (\ref{el3}), and let $y_m(\omega)=\mathop
{\inf }\limits_{n \ge m} x_n$. Since for all $\omega \in
\bigcap\limits_{m 1}^\infty  \{\omega :\mathop {\inf }\limits_{n
\ge m} x_n \le 0 \}$, $y_m(\omega) \le 0$ for all $m$, then by
definition $\mathop {\sup }\limits_m y_m(\omega) \le 0$ as well.
On the other hand if $\mathop {\sup }\limits_m y_m(\omega) \le 0$,
then for all $m$, $y_m(\omega) \le 0$. Hence we have the following
set equality
\begin{equation}
\bigcap\limits_{m 1}^\infty  \{\omega :\mathop {\inf }\limits_{n
\ge m} x_n \le 0 \} =  \{\omega :\mathop {\sup }\limits_m \mathop
{\inf }\limits_{n \ge m} x_n \le 0 \}.
\end{equation}
Rewriting (\ref{el3}) we have
\begin{equation}
\bigcap\limits_{m = 1}^\infty \bigcup\limits_{n = m}^\infty A_n
\subseteq \{\omega :\mathop {\sup }\limits_m \mathop {\inf
}\limits_{n \ge m} x_n \le 0 \}=\{ \omega : \mathop {\liminf
}\limits_{n \to \infty} x_n(\omega) \le 0\}, \label{pp}
\end{equation}
where the equality on the R.H.S. of (\ref{pp}) follows from the
definition of $\mathop {\liminf }\limits_{n \to \infty}(\cdot)$ of
a sequence $x_n$. Also by definition, $\bigcap\limits_{m 1}^\infty
\bigcup\limits_{n = m}^\infty A_n=\mathop {\limsup }\limits_{n \to
\infty} A_n$. Hence, (see, \eg, \cite{chung}, p. 67)
\begin{equation}
\mathop {\limsup }\limits_{n \to \infty} A_n =\{\omega:
x_n(\omega) \le 0\ i.o. \} \subseteq \{ \omega : \mathop {\liminf
}\limits_{n \to \infty} x_n(\omega) \le 0\}.
\end{equation}
Due to the monotonicity of the probability measure, the lemma
follows.
\end{proof}

\vspace{0.3in}

%
%
\appendix{\large {\bf Appendix C}}

Let $D$ be an \emph{infinite} order non-symmetrical half-plane
support defined as in (\ref{d}) and let $D(k,l)$ be a \emph{finite}
order non-symmetrical half-plane support, defined by

\begin{equation}
D(k,l)=\left\{(i,j)\in \mathbb{Z}^2: i=0, 0 \leq j \leq l
\right\}\cup \left\{(i,j)\in \mathbb{Z}^2: 0<i\leq k, -l \leq j \leq
l \right\}
\end{equation}

Let the field $\{ w(n,m) \}$ be defined as in (\ref{e3}), 
%
and the field $\{u(n,m)\}$ is an i.i.d. real valued zero-mean random
field with finite second order moment, $\sigma^2$. The sequence
${a(i,j)}$ is a square summable deterministic sequence,

\begin{equation}
\sum_{(r,s) \in D}a^2(r,s) < \infty. \label{lemma1}
\end{equation}

The next lemma is an extension of a lemma originally proposed by
Hannan, \cite{hannan} for the case of 1-D signals. Similar result
can be found in \cite{kundu}, Lemma 2, but with only a partial
proof. Since this lemma is crucial for our work we will prove it
here.

\begin{lemma}
\label{XXX}
\begin{equation}
\mathop {\sup }\limits_{\omega,\upsilon } \left |
\frac{1}{NM}\sum\limits_{n = 0}^{N-1} \sum\limits_{m=0}^{M-1} w(n,m)
\cos{(\omega n +\nu m)}
  \right | \to 0 \ \mbox{a.s.} \ \mbox{as} \ \Psi (N,M) \to \infty
\end{equation}
\end{lemma}
\begin{proof}

First, it is easy to see that,
\begin{eqnarray}
&&\mathop {\sup }\limits_{\omega,\upsilon } \left |
\frac{1}{NM}\sum\limits_{n = 0}^{N-1} \sum\limits_{m=0}^{M-1} w(n,m)
\cos{(\omega n +\nu m)}
  \right |\leq  \nonumber \\
  &&\mathop {\sup }\limits_{\omega,\upsilon } \left |
\frac{1}{2NM}\sum\limits_{n = 0}^{N-1} \sum\limits_{m=0}^{M-1}
w(n,m) e^{j(\omega n +\nu m)}
  \right | + \mathop {\sup }\limits_{\omega,\upsilon } \left |
\frac{1}{2NM}\sum\limits_{n = 0}^{N-1} \sum\limits_{m=0}^{M-1}
w(n,m) e^{-j(\omega n +\nu m)}
  \right |.
\end{eqnarray}

Hence it is sufficient to prove the lemma for exponentials, \ie, we
wish to prove that
\begin{equation}
\mathop {\sup }\limits_{\omega,\upsilon } \left |
\frac{1}{NM}\sum\limits_{n = 0}^{N-1} \sum\limits_{m=0}^{M-1} w(n,m)
e^{j(\omega n +\nu m)}
  \right | \to 0 \ \mbox{a.s.} \ \mbox{as} \ \Psi (N,M) \to \infty
\end{equation}

Define the set $D(k,l)^C=D\setminus D(k,l)$. Then,
\begin{equation}
w(n,m)=\sum_{ D(k,l)} a(r,s) u(n-r,m-s) + \sum_{ D(k,l)^C} a(r,s)
u(n-r,m-s)= v(n,m)+z(n,m). \label{lemma2a}
\end{equation}

Then,
\begin{eqnarray}
\mathop {\sup }\limits_{\omega,\upsilon } \left |
\frac{1}{NM}\sum\limits_{n = 0}^{N-1} \sum\limits_{m=0}^{M-1} z(n,m)
e^{j(\omega n +\nu m)}
  \right |\leq\frac{1}{NM}\sum\limits_{n = 0}^{N-1}
  \sum\limits_{m=0}^{M-1}|z(n,m)| \leq \left\{\frac{1}{NM}\sum\limits_{n = 0}^{N-1}
  \sum\limits_{m=0}^{M-1}z^2(n,m)\right\}^{\frac{1}{2}}.
\end{eqnarray}

By the SLLN, the R.H.S. of the last inequality convergence, almost
surely, to
\begin{equation}
{ E}[z(0,0)^2]^{\frac{1}{2}}=\left(\sigma^2\sum_{ D(k,l)^C}
a(r,s)^2\right)^{\frac{1}{2}},
\end{equation}
which due to (\ref{lemma1}) may be made arbitrary small by taking
$k$ and $l$ sufficiently large.

Hence it is sufficient to prove the lemma with $w(n,m)$ replaced by
$v(n,m)$.

\begin{eqnarray}
\hspace{-.3in} \mathop {\sup }\limits_{\omega,\upsilon } \left |
\frac{1}{NM}\sum\limits_{n = 0}^{N-1} \sum\limits_{m=0}^{M-1} v(n,m)
e^{j(\omega n +\nu m)}
  \right |\leq \sum_{ D(k,l)} |a(r,s)|\mathop {\sup }\limits_{\omega,\upsilon } \left |
\frac{1}{NM}\sum\limits_{n = 0}^{N-1} \sum\limits_{m=0}^{M-1}
u(n-r,m-s) e^{j(\omega n +\nu m)}\right |.
\end{eqnarray}

Since the summation is finite and $\{u(n,m)\}$ is i.i.d., it is
sufficient to prove the lemma with $w(n,m)$ replaced by $u(n,m)$.
Thus, we consider the mean square of the discussed supremum

\begin{eqnarray}
&&
{E}\left[\mathop {\sup }\limits_{\omega,\upsilon } \left |
\frac{1}{NM}\sum\limits_{n = 0}^{N-1} \sum\limits_{m=0}^{M-1} u(n,m)
e^{j(\omega n +\nu m)}
  \right |^2 \right] \nonumber \\
  &&={E}\left[\mathop {\sup }\limits_{\omega,\upsilon }
\frac{1}{(NM)^2}\sum\limits_{n = 0}^{N-1}
\sum\limits_{m=0}^{M-1}\sum\limits_{k = 0}^{N-1}
\sum\limits_{l=0}^{M-1} u(n,m)u(k,l) e^{j(\omega (n-k) +\nu (m-l))}
   \right].
   \label{lemma2}
\end{eqnarray}

By letting,

\begin{equation}
\label{eq23}
\begin{array}{l}
 n - k = p, \\
 m - l = r,
 \end{array}
\end{equation}

substitute,

\begin{equation}
\begin{array}{l}
\label{eq24} \sum\limits_{n = 0}^{N - 1} \sum\limits_{k = 0}^{N - 1}
= \sum\limits_{\vert
p\vert < N} \sum\limits_{n \in S_N },\\
\sum\limits_{m = 0}^{M - 1} \sum\limits_{l = 0}^{M - 1} =
\sum\limits_{\vert r\vert < M} \sum\limits_{m \in S_M },
 \end{array}
\end{equation}

where,
\begin{equation}
\label{eq25}
\begin{array}{l}
S_N = \{n \in \mathbb{Z} :\max (0,p) \le n \le \min (N-1,p + N-1)\},\\
S_M = \{m \in \mathbb{Z} :\max (0,r) \le m \le \min (M-1,r + M-1)\},
 \end{array}
\end{equation}

and,

\begin{equation}
\label{eq26}
\begin{array}{l}
\sum\limits_{n \in S_N } 1 = \left\{ {{\begin{array}{*{20}c}
 {N - p,\quad p \ge 0} \hfill \\
 {N + p,\quad p < 0} \hfill \\
\end{array} } = N - \vert p\vert } \right. ,\\
\sum\limits_{m \in S_M } 1 = \left\{ {{\begin{array}{*{20}c}
 {M - r,\quad r \ge 0} \hfill \\
 {M + r,\quad r < 0} \hfill \\
\end{array} } = M - \vert r\vert } \right..
\end{array}
\end{equation}

Hence, rewriting (\ref{lemma2}) we have

\begin{eqnarray}
&&\frac{1}{(NM)^2}{E}\left[\mathop {\sup }\limits_{\omega,\upsilon }
\sum\limits_{\vert p\vert < N}\sum\limits_{\vert r\vert < M}
\sum\limits_{n \in S_N }\sum\limits_{m \in S_M }u(n,m)u(n-p,m-r)
e^{j(\omega p +\nu r)}
\right]\nonumber \\
&&=\frac{1}{(NM)^2}{E}\left[\sum\limits_{n =0 }^{N-1}\sum\limits_{m
=0 }^{M-1}u(n,m)^2+ \mathop {\sup }\limits_{\omega,\upsilon }
 \mathop{\sum\limits_{\vert
p\vert < N}}\limits_{p\neq0}\mathop{\sum\limits_{\vert r\vert <
M}}\limits_{r\neq0} \sum\limits_{n \in S_N }\sum\limits_{m \in S_M
}u(n,m)u(n-p,m-r) e^{j(\omega p +\nu r)}
   \right] \nonumber \\
   &&\leq \frac{1}{(NM)^2}\left\{NM\sigma^2+
 \mathop{\sum\limits_{\vert
p\vert < N}}\limits_{p\neq0}\mathop{\sum\limits_{\vert r\vert <
M}}\limits_{r\neq0}{E}\left[ \left| \sum\limits_{n \in S_N
}\sum\limits_{m \in S_M } u(n,m)u(n-p,m-r)\right|\right]\right\},
\label{lemma3}
\end{eqnarray}
where in the first equality we split up the sum into the squared
term and the remainder, and then employ the triangular inequality.

Let us investigate the second term on the R.H.S. of (\ref{lemma3}).
From the Cauchy-Schwartz inequality, for any r.v. $x$, ${E}\left[
\left|x\right|\right]\leq{E}\left[
\left|x\right|^2\right]^{\frac{1}{2}}$, hence
\begin{eqnarray}
&&{E}\left[ \left| \sum\limits_{n \in S_N }\sum\limits_{m \in S_M }
u(n,m)u(n-p,m-r)\right|\right]\leq{E}\left[ \left| \sum\limits_{n
\in S_N }\sum\limits_{m \in S_M }
u(n,m)u(n-p,m-r)\right|^2\right]^{\frac{1}{2}}\nonumber \\
&&=\left(\sum\limits_{n \in S_N }\sum\limits_{m \in S_M
}\sum\limits_{n' \in S_N }\sum\limits_{m' \in S_M
}{E}[u(n,m)u(n-p,m-r)u(n',m')u(n'-p,m'-r)] \right)^{\frac{1}{2}}\nonumber \\
&&=\left(\sum\limits_{n \in S_N }\sum\limits_{m \in
S_M}\sigma^4\right)^{\frac{1}{2}}=\sigma^2(N-|p|)^{\frac{1}{2}}(M-|r|)^{\frac{1}{2}}.
\end{eqnarray}
which follows from the observation that for $p,r\neq0$, the fourth
order moment of the field $\{u(n,m)\}$ equals zero for all $n\neq
n'$ or $m\neq m'$.

Hence we can finally write

\begin{eqnarray}
&&{E}\left[\mathop {\sup }\limits_{\omega,\upsilon } \left |
\frac{1}{NM}\sum\limits_{n = 0}^{N-1} \sum\limits_{m=0}^{M-1} u(n,m)
e^{j(\omega n +\nu m)}
  \right |^2 \right]\nonumber \\
  &&\leq\frac{1}{(NM)^2}\left\{NM\sigma^2+
 \mathop{\sum\limits_{\vert
p\vert < N}}\limits_{p\neq0}\mathop{\sum\limits_{\vert r\vert <
M}}\limits_{r\neq0}\sigma^2(N-|p|)^{\frac{1}{2}}(M-|r|)^{\frac{1}{2}}\right\}\nonumber
\\
&& \leq\frac{\sigma^2}{(NM)^2}\{NM+4(NM)^{\frac{3}{2}}\} \leq
\frac{K}{(NM)^{\frac{1}{2}}}=O(N^{-\frac{1}{2}}M^{-\frac{1}{2}}).
\label{lemma4}
\end{eqnarray}
where $K$ some finite positive constant.

Now following the ideas of Doob, \cite{Do}( ch. X, 6), let $R$ and
$S$ be some positive integers such that $N>R^\delta$, and
$M>S^\delta$, for $\delta>2$. Hence, for any such choice of $N$ and
$M$, from (\ref{lemma4}),
\begin{eqnarray}
&&{E}\left[\mathop {\sup }\limits_{\omega,\upsilon } \left |
\frac{1}{NM}\sum\limits_{n = 0}^{N-1} \sum\limits_{m=0}^{M-1} u(n,m)
e^{j(\omega n +\nu m)}
  \right |^2 \right]\leq
\frac{K}{(RS)^{\frac{\delta}{2}}}. \label{lemma5}
\end{eqnarray}

Hence, if we take $N=N(R)$ and $M=M(S)$ to be the smallest integers
not smaller then $R^\delta$ and $S^\delta$, respectively, then
(\ref{lemma5}) still holds.

Hence, by Chebyshev inequality for every $\epsilon>0$
\begin{eqnarray}
&&P\left(\mathop {\sup }\limits_{\omega,\upsilon } \left |
\frac{1}{N(R)M(S)}\sum\limits_{n = 0}^{N(R)-1}
\sum\limits_{m=0}^{M(S)-1} u(n,m) e^{j(\omega n +\nu m)}
  \right |\geq \epsilon\right)\nonumber \\
  &&\leq \frac{{E}\left[\mathop {\sup }\limits_{\omega,\upsilon }
\left | \frac{1}{N(R)M(S)}\sum\limits_{n = 0}^{N(R)-1}
\sum\limits_{m=0}^{M(S)-1} u(n,m) e^{j(\omega n +\nu m)}
  \right |^2 \right]}{\epsilon^2}\leq\frac{K}{\epsilon^2(RS)^{\frac{\delta}{2}}}
\end{eqnarray}
and then since $\delta>2$
\begin{eqnarray}
\sum_{R=1}^{\infty}\sum_{S=1}^{\infty}P\left(\mathop {\sup
}\limits_{\omega,\upsilon } \left | \frac{1}{N(R)M(S)}\sum\limits_{n
= 0}^{N(R)-1} \sum\limits_{m=0}^{M(S)-1} u(n,m) e^{j(\omega n +\nu
m)}
  \right
  |>\epsilon\right)\leq\sum_{R=1}^{\infty}\sum_{S=1}^{\infty}\frac{K}{\epsilon^2(RS)^{\frac{\delta}{2}}}<\infty.
\end{eqnarray}
Hence, by the Borel-Cantelly lemma,

\begin{equation}
\mathop {\sup }\limits_{\omega,\upsilon } \left |
\frac{1}{N(R)M(S)}\sum\limits_{n = 0}^{N(R)-1}
\sum\limits_{m=0}^{M(S)-1} u(n,m) e^{j(\omega n +\nu m)}
  \right| \to 0 \ \mbox{a.s.} \ \mbox{as} \ \Psi (R,S) \to \infty.
  \label{lemma6}
\end{equation}

Now,
\begin{eqnarray}
&&\mathop{\mathop{\sup }\limits_{N(R)\leq N \leq
N(R+1)}}\limits_{M(S)\leq M \leq M(S+1)} \mathop {\sup
}\limits_{\omega,\upsilon } \left | \frac{1}{NM}\sum\limits_{n =
0}^{N-1} \sum\limits_{m=0}^{M-1} u(n,m) e^{j(\omega n +\nu m)}
  -\frac{1}{NM}\sum\limits_{n = 0}^{N(R)-1}
\sum\limits_{m=0}^{M(S)-1} u(n,m) e^{j(\omega n +\nu m)}\right |
\nonumber\\
&&\leq\mathop{\mathop{\sup }\limits_{N(R)\leq N \leq
N(R+1)}}\limits_{M(S)\leq M \leq M(S+1)} \mathop {\sup
}\limits_{\omega,\upsilon } \frac{1}{NM}\left|\sum\limits_{n =
0}^{N(R)-1} \sum\limits_{m=M(S)}^{M-1} u(n,m)e^{j(\omega n +\nu
m)}\right|\nonumber \\
&&+\mathop{\mathop{\sup }\limits_{N(R)\leq N \leq
N(R+1)}}\limits_{M(S)\leq M \leq M(S+1)} \mathop {\sup
}\limits_{\omega,\upsilon } \frac{1}{NM}\left|\sum\limits_{n =
N(R)}^{N-1} \sum\limits_{m=0}^{M(S)-1} u(n,m)e^{j(\omega n +\nu
m)}\right| \nonumber \\
&& +\mathop{\mathop{\sup }\limits_{N(R)\leq N \leq
N(R+1)}}\limits_{M(S)\leq M \leq M(S+1)} \mathop {\sup
}\limits_{\omega,\upsilon } \frac{1}{NM}\left|\sum\limits_{n =
N(R)}^{N-1} \sum\limits_{m=M(S)}^{M-1} u(n,m)e^{j(\omega n +\nu
m)}\right| =I_1+I_2+I_3.
\end{eqnarray}

Consider the first term in the previous equation. Using the
triangular inequality
\begin{equation}
I_1\leq \frac{1}{M(S)}\sum\limits_{m=M(S)}^{M(S+1)-1} \left (
\mathop {\sup }\limits_{\omega} \frac{1}{N(R)}\left|\sum\limits_{n =
0}^{N(R)-1} u(n,m)e^{j\omega n}\right| \right).
\end{equation}
Let
\begin{equation}
\tilde u (m)= \mathop {\sup }\limits_{\omega}
\frac{1}{N(R)}\left|\sum\limits_{n = 0}^{N(R)-1} u(n,m)e^{j\omega
n}\right|.
\end{equation}
Since $\{u(n,m)\}$ is i.i.d., it is clear that $\{\tilde u (m)\}$ is
an i.i.d. sequence of random variables. Moreover, from \cite{hannan}
(or by repeating the derivation in (\ref{lemma2a})-(\ref{lemma5})
for the process $u(n,m)$ with a fixed $m$) we have
\begin{eqnarray}
&&{E}\left[\tilde u (m)^2\right]= {E}\left[ \mathop {\sup
}\limits_{\omega} \frac{1}{N(R)}\left|\sum\limits_{n = 0}^{N(R)-1}
u(n,m)e^{j\omega n}\right|^2\right]\leq
\frac{K_1}{R^{\frac{\delta}{2}}}.
\end{eqnarray}
 Taking the mean of the square of the $I_1$ we have
\begin{eqnarray}
&&{E}\left[|I_1|^2\right]\leq \frac{1}{M(S)^2}
\sum\limits_{m=M(S)}^{M(S+1)-1}\sum\limits_{m'=M(S)}^{M(S+1)-1}
{E}\left[\tilde u(m)\tilde u(m')\right] \nonumber \\
&&\leq \frac{1}{M(S)^2}
\sum\limits_{m=M(S)}^{M(S+1)-1}\sum\limits_{m'=M(S)}^{M(S+1)-1}
{E}\left[\tilde u(m)^2\right]^{\frac{1}{2}}{E}\left[\tilde u(m')^2\right]^{\frac{1}{2}} \nonumber \\
&&\leq\frac{K_1(M(S+1)-1-M(S))^2}{R^{\frac{\delta}{2}}M(S)^2}\leq
\frac{K}{R^{\frac{\delta}{2}}S^2}.
\end{eqnarray}

 Using once again the Chebyshev inequality and
the Borel-Cantelli lemma we have that $I_1\to 0$ a.s. as
$\Psi(R,S)\to \infty$. Repeating the same consideration for $I_2$ we
have that $I_2\to 0$ a.s. as $\Psi(R,S)\to \infty$. Finally, for
$I_3$ we have

\begin{eqnarray}
&&{E}[|I_3|^2]\leq {E}\left[ \left|\frac{1}{N(R)M(S)}\sum\limits_{n
= N(R)}^{N(R+1)-1} \sum\limits_{m=M(S)}^{M(R+1)-1} |u(n,m)|\right|^2
\right]\nonumber\\
&& = \frac{1}{(N(R)M(S))^2}\sum\limits_{n = N(R)}^{N(R+1)-1}
\sum\limits_{m=M(S)}^{M(S+1)-1}\sum\limits_{n' = N(R)}^{N(R+1)-1}
\sum\limits_{m'=M(S)}^{M(S+1)-1}{E}[|u(n,m)u(n',m')|] \nonumber \\
&&\leq
\sigma^2\frac{(N(R+1)-1-N(R))^2(M(S+1)-1-M(S))^2}{(N(R)M(S))^2}
 \leq \frac{K}{(RS)^2}.
\end{eqnarray}

Using again the Chebyshev inequality and the Borel-Cantelli lemma we
have that $I_3\to 0$ a.s. as $\Psi(R,S)\to \infty$.

Finally, we have that
\begin{equation}
\mathop {\sup }\limits_{\omega,\upsilon } \left |
\frac{1}{NM}\sum\limits_{n = 0}^{N-1} \sum\limits_{m=0}^{M-1} u(n,m)
e^{j(\omega n +\nu m)}
  -\frac{1}{NM}\sum\limits_{n = 0}^{N(R)-1}
\sum\limits_{m=0}^{M(S)-1} u(n,m) e^{j(\omega n +\nu m)}\right | \to
0 \ \mbox{a.s.}
\end{equation}
for all $N(R)\leq N <N(R+1)$ and $M(S)\leq M <M(S+1)$, as
$\Psi(R,S)\to \infty$, and hence as $\Psi(N,M)\to \infty$,

Since $\frac{N(R)}{N(R+1)}\to1$ and $\frac{M(S)}{M(S+1)}\to1$ as
$\Psi(R,S)\to \infty$ we can replace $\frac{1}{NM}$ in the second
term by $\frac{1}{N(R)M(S)}$. Therefore, we have
\begin{equation}
\mathop {\sup }\limits_{\omega,\upsilon } \left |
\frac{1}{NM}\sum\limits_{n = 0}^{N-1} \sum\limits_{m=0}^{M-1} u(n,m)
e^{j(\omega n +\nu m)}
  -\frac{1}{N(R)M(S)}\sum\limits_{n = 0}^{N(R)-1}
\sum\limits_{m=0}^{M(S)-1} u(n,m) e^{j(\omega n +\nu m)}\right | \to
0 \ \mbox{a.s.} \label{last}
\end{equation}
From (\ref{last}) and (\ref{lemma6}) the lemma  follows.

\end{proof}

%
%

\end{document}